\documentclass[10pt, doublecolumn, twoside]{IEEEtran}
\usepackage{amsmath,epsfig}
\usepackage{enumerate}

\usepackage{cite}
\usepackage{graphicx,subfigure}
\usepackage{amssymb}
\usepackage{multirow}
\begin{document}

\newtheorem{lem}{Lemma}
\newtheorem{prop}{Proposition}
\newtheorem{cor}{Corollary}
\newtheorem{remark}{Remark}
\newtheorem{defin}{Definition}
\newtheorem{thm}{Theorem}

\newcounter{MYtempeqncnt}

\title{Multi-user Linear Precoding for Multi-polarized Massive MIMO System under Imperfect CSIT}

\author{Jaehyun Park,~\IEEEmembership{Member,~IEEE,}
     Bruno Clerckx,~\IEEEmembership{Member,~IEEE,}
\thanks{A part of this work was published in the 22nd European Signal Processing Conference EUSIPCO 2014.}
\thanks{J. Park is with the Department of Electronic Engineering, Pukyong National University, Republic of Korea. B. Clerckx is with the Department of Electrical and Electronic Engineering, Imperial College London, United Kingdom and also with School
of Electrical Engineering, Korea University (e-mail:jaehyun@pknu.ac.kr, b.clerckx@imperial.ac.uk). B. Clerckx is the corresponding author.}\thanks{This work is partially supported by Huawei Technologies Co., Ltd.}
}

%\author{\authorblockN{Jaehyun Park\authorrefmark{1}, Byung Jang Jeong\authorrefmark{1}, and Yunju Park\authorrefmark{2}}
%\authorblockA{\authorrefmark{1}Broadcasting and
%Telecommunications Convergence Research Laboratory\\
%Electronics and Telecommunications Research Institute (ETRI),
%Daejeon, Korea\\
%Email: \{jhpark00, shwang, bjjeong\}@etri.re.kr}
%\authorblockA{\authorrefmark{2}Department of Mathematics and Computer Science\\
%Korea Science Academy of Korea Advanced Institute of Science and Technology (KAIST), Busan, Korea\\
%Email: yunjupark@kaist.ac.kr}}
% make the title area
\maketitle

\begin{abstract}
The space limitation and the channel acquisition prevent Massive MIMO from being easily deployed in a practical setup. Motivated by current deployments of LTE-Advanced, the use of multi-polarized antenna elements can be an efficient solution to address the space constraint. Furthermore, the dual-structured precoding, in which a preprocessing based on the spatial correlation and a subsequent linear precoding based on the short-term channel state information at the transmitter (CSIT) are concatenated, can reduce the feedback overhead efficiently. By grouping and preprocessing spatially correlated mobile stations (MSs), the dimension of the precoding signal space is reduced and the corresponding short-term CSIT dimension is reduced. In this paper, to reduce the feedback overhead further, we propose a dual-structured multi-user linear precoding, in which the subgrouping method based on co-polarization is additionally applied to the spatially grouped MSs in the preprocessing stage. Furthermore, under imperfect CSIT, the proposed scheme is asymptotically analyzed based on random matrix theory. By investigating the behavior of the asymptotic performance, we also propose a new dual-structured precoding in which the precoding mode is switched between two dual-structured precoding strategies with i) the preprocessing based only on the spatial correlation and ii) the preprocessing based on both the spatial correlation and polarization. Finally, we extend it to 3D dual-structured precoding.
\end{abstract}


\section{Introduction}
\label{sec:intro}
By deploying a large number of antenna elements at the base station (BS), high spectral efficiencies can be achieved. Furthermore, because low-complexity linear precoding schemes can be efficiently exploited to serve multiple mobile stations (MSs) simultaneously, Massive MIMO plays a key role in beyond 4G cellular networks (\cite{Marzetta, Marzetta2, Huh, HoydisDebbah} and references therein).

Assuming large scale arrays, several linear single-user and multi-user precoding schemes are asymptotically analyzed by using random matrix theory in \cite{MohLarsson, MohLarsson2, Wagner1, Caire, ZhangWen}. In \cite{MohLarsson}, single-user beamforming in MISO with a large number of transmit antenna elements has been analyzed under the per-antenna constant-envelope constraints and its extension to multi-user MIMO is treated in \cite{MohLarsson2}. In \cite{Wagner1}, by utilizing a Stieltjes transform of random positive semi-definite matrices, the asymptotic signal-to-interference-and-noise ratio (SINR) of linear precoding in a correlated Massive MISO broadcasting channel has been derived under imperfect channel state information at the transmitter (CSIT). In \cite{ZhangWen}, by considering multi-cell downlink with massive transmit antenna elements, the asymptotic SINR is analyzed by using random matrix theory which gives an insight about the cooperative transmission strategy of BSs.

However, in FDD, where channel reciprocity is not exploitable, the multi-antenna channel acquisition at the transmitter prevents Massive MIMO from being easily deployed. To resolve the channel acquisition burden at BS, a dual structured precoding, in which a preprocessing based on the long-term CSIT (mainly, spatial correlation) and a subsequent linear precoding based on the short-term CSIT (that generally has lower dimension than the number of transmit antenna elements) are concatenated, can be exploited to reduce the feedback overhead efficiently \cite{Clerckx_book, Caire}. Note that the long-term CSI is slowly-varying and can be obtained accurately with a low feedback overhead. Because of the low feedback overhead and the attractive performance, the dual structured precoding has been also considered in the 4G and beyond 4G wireless standards \cite{LimYooClerckx, Clerckx_book}. In \cite{Caire}, when MSs are clustered as several spatially correlated groups, joint spatial division and multiplexing scheme has been proposed in which the precoding matrix is composed of the {\it{prebeamforming}} matrix based on the spatial correlation and the {\it{classical precoder}} based on short-term CSIT. Furthermore, its performance is asymptotically analyzed for a large number of transmit antenna elements.

Another challenge of the Massive MIMO system is the antenna space limitation. An increasing number of antenna elements is difficult to be packed in a limited space and if it can be deployed, the high spatial correlation and the mutual coupling among the antennas elements may cause some system performance degradation, especially for a small numbers of active MSs \cite{BrunoKimCon,ClerckxCraeye}. The multi-polarized antenna elements can be one solution to alleviate the space constraint \cite{KimBruno,OestgesBruno}. The multi-polarized antenna systems have been investigated under various communication scenarios including the picocell/microcell \cite{Sulonen}, indoor/outdoor \cite{DongHeath}, and the line of sight (LOS)/non LOS (NLOS) \cite{OestgesBruno} environments. However, despite the importance of polarized antennas in practical deployments, the Massive MIMO system with multi-polarized antenna elements has not been addressed so far together with the multi-user linear precoding. Note that due to space constraints, closely spaced dual-polarized antennas is considered as the first priority deployment scenario for MIMO in LTE-A and is therefore likely to remain so as the number of antennas at the base station increases \cite{Clerckx_book, LimYooClerckx, Boccardi}.

In this paper, we first model the Massive multi-user MIMO system, where the BS is equipped with a large number of multi-polarized antenna elements. For simplicity and practical issue of MS, we consider that MSs are equipped with a single single-polarized antenna element. We then present the dual structured precoding based on long-term/short-term CSIT. As done in the Massive MIMO system with a single-polarized antenna element \cite{Caire}, by grouping the spatially correlated MSs and multiplying the channel matrix of the grouped MSs with the same preprocessing matrix based on spatial correlation, the dimension of the precoding signal space is reduced and the corresponding short-term CSIT dimension can also be reduced from the Karhunen-Loeve transform \cite{Rao_book}. Considering the multi-polarized Massive multi-user MIMO system for the first time, the contributions of this paper are listed below:
\begin{itemize}
\item To reduce the feedback overhead further, we first propose a dual structured linear precoding, in which the subgrouping method based on the polarization is additionally applied to the spatially grouped MSs in the preprocessing stage. That is, by subgrouping co-polarized MSs in each group, we let the MS report the CSI from the transmit antenna elements having the same polarization as its polarization. Then, the MS can further reduce the short-term CSI feedback overhead, compared to the case with the conventional preprocessing based only on the spatial correlation.

\item Under the imperfect CSIT, two different dual structured precodings with preprocessing of i) grouping based only on the spatial correlation (i.e., spatial grouping) and ii) subgrouping based on both the spatial correlation and polarization are asymptotically analyzed based on random matrix theory with a large dimension \cite{Wagner1}. Because, in this paper, the polarized Massive MU-MIMO channel is considered, the asymptotic inter/intra interferences are evaluated over the polarization domain as well as the spatial domain, which addresses a more general (polarized) channel environment compared to \cite{Wagner1, Caire}. Accordingly, the asymptotic performance can be further analyzed in terms of both the polarization and the spatial correlation and therefore, we can understand the performance behavior of dual structured precoding with respect to the long-term CSIT.
\item We then propose a new dual precoding method to switch the precoding mode between two dual structured precoding strategies relying on i) the spatial grouping only and ii) the subgrouping based on both the spatial correlation and polarization.
\item Motivated by 3D beamforming \cite{Caire, LiJi}, we extend the design to the 3D dual structured precoding in which the spatial correlation depends on both azimuth and elevation angles.
\item Finally, we also discuss how we can modify the proposed precoding mode switching scheme when the polarization at the BS and MS is mismatched (due to e.g. random MS orientation).

\end{itemize}
Note that, from the asymptotic results, we can find that even though the proposed dual precoding using the subgrouping can reduce the feedback overhead, its performance can be affected by the cross-polar discrimination (XPD) parameter. Here, XPD refers to the long-term statistics of the antenna elements and channel depolarization that measures the ability to distinguish the orthogonal polarization. That is, under the same feedback overhead, the dual precoding with subgrouping can utilize more accurate CSIT on half of the array, compared with precoding with spatial grouping, but exhibits performance more sensitive to the XPD. The precoding with spatial grouping can only utilize less accurate CSIT, but its performance is not affected by the XPD. Accordingly, we identify the region where the dual precoding with subgrouping outperforms that with spatial grouping. The region depends on the XPD, the spatial correlation, and the short-term CSIT quality, which motivated us to develop the new dual precoding method.

The rest of this paper is organized as follows. In Section
\ref{sec:systemmodel}, we introduce the Massive MIMO system model with multiple multi-polarized antenna elements at the BS and a (either vertically or horizontally) single polarized antenna element at the multiple MSs. In Section \ref{sec:dualprecoding}, we discuss the dual structured precoding based on the long-term/short-term CSIT. In Section \ref{sec:Asymptotic}, we investigate the asymptotic performance of the dual precoding schemes and their behavior over the XPD parameter. Based on the asymptotic results, in
Section \ref{sec:newprecoding}, we propose a new dual structured precoding/feedback. In Section
\ref{sec:simulation}, we provide several discussion and simulation
results, respectively, and in Section \ref{sec:conc} we give our
conclusions.

Throughout the paper, matrices and vectors are represented by bold
capital letters and bold lower-case letters, respectively. The
notations $({\bf A})^{T} $, $({\bf A})^{H} $, $({\bf A})_i$,
$[{\bf A}]_i$, $tr({\bf A})$, and $\det({\bf A})$ denote the
transpose, conjugate transpose, the $i$th row, the $i$th
column, the trace, and the determinant of a matrix ${\bf A}$,
respectively. In addition, $[{\bf A}]_{i:j}$ (resp., $({\bf A})_{i:j}$) denotes the submatrix from the $i$th column (resp., row) to the $j$th column (resp., row) of ${\bf A}$. The matrix norm $\|{\bf A}\|$ and the vector norm $\|{\bf a}\|$ denote the 2-norms of a matrix ${\bf A}$ and a vector ${\bf a}$, respectively. In addition, ${\bf A} \succeq 0$ means that a matrix ${\bf A}$ is positive
semi-definite, $\otimes$ denotes the Kronecker product, and $\odot$ denotes the Hadamard product. The operation $E_{g}[A_g]$ means the average of $A_g$ over index g. Finally, ${\bf I}_{M}$, ${\bf 1}_{M\times N}$, and ${\bf 0}_{M\times N}$ denote the $M \times M$ identity matrix, the $M$ by $N$ matrix with all $1$ entries, and the $M$ by $N$ matrix with all $0$ entries, respectively.

\section{System model}
\label{sec:systemmodel}

\begin{figure}
\begin{center}
\begin{tabular}{c}
\includegraphics[height=4.5cm]{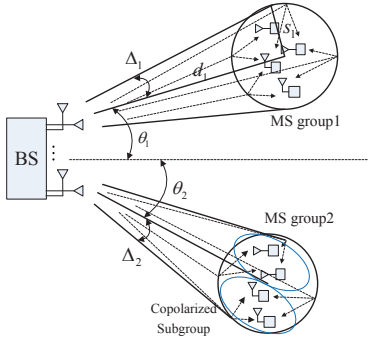}
\end{tabular}
\end{center}
\caption[Fig_system_block]
%>>>> use \label inside caption to get Fig. number with \ref{}
{ \label{Fig_system_block} Block diagram of a dual-polarized multi-user downlink system.}
\end{figure}

We consider a single-cell downlink system with one BS with $M$ polarized antenna elements and $N$ active MSs, each with a single polarized antenna element, where $M$ and $N$ are assumed to be even numbers. As in Fig \ref{Fig_system_block}, the BS has $\frac{M}{2}$ pairs of co-located vertically/horizontally polarized antenna elements and the MSs have a single antenna element with either vertical or horizontal polarization\footnote{Throughout the paper, we consider the dual-polarized antenna elements at the BS for ease of explanation, but our approach can be extended to the tri-polarized case without difficulty. In addition, all the antenna elements are perfectly aligned with either vertical or horizontal polarization. In Section \ref{ssec:polmismatch}, we discuss the polarization mismatch due to the random orientation of MSs.}. Furthermore, since human activity is usually confined in small clustered regions such as buildings, locations of MSs tend to be spatially clustered, e.g., $G$ groups. Then, the received signal ${\bf y}_g\in \mathbb{C}^{N_g}$ of the $g$th group with an assumption of flat-fading channel is given by
\begin{eqnarray}\label{Sys_1}
{\bf y}_g =\left[\begin{array}{c}{\bf y}_g^v\\{\bf y}_g^h\end{array}\right]= {\bf H}_g^H {\bf x} + {\bf n}_g,
\end{eqnarray}
where ${\bf y}_g^v$ and ${\bf y}_g^h$ are the received signal for MSs with vertical and horizontal polarization, respectively, and ${\bf n}_g=\left[\begin{array}{c}{\bf n}_g^v\\{\bf n}_g^h\end{array}\right]$ is a zero-mean complex Gaussian noise vector having a covariance matrix ${\bf I}_{N_g}$, denoted as ${\bf n}_g \sim CN({\bf 0}, {\bf I}_{N_g})$. Here, $N_g$ denotes the number of MSs in the $g$th group. For simplicity, it is assumed that $N_1=...=N_G = \bar N$, where $\bar N$ is even, and both ${\bf y}_g^v$ and ${\bf y}_g^h$ are $\frac{\bar N}{2} \times 1$ vectors. The channel of the $k$th MS in the $g$th group is then given as ${\bf h}_{gk} = [{\bf H}_g]_k$. The $M\times 1$ vector, ${\bf x}$, is the linearly precoded transmit signal expressed as
\begin{eqnarray}\label{Sys_2}
{\bf x} =\sum_{g=1}^G {\bf V}_g{\bf d}_g, \quad {\bf V}_g \in \mathbb{C}^{M\times \bar N},
\end{eqnarray}
where ${\bf V}_g$ and ${\bf d}_g=\left[\begin{array}{c}{\bf d}_g^v\\{\bf d}_g^h\end{array}\right]$ are the linear precoding matrix and the data symbol vector for the MSs in the $g$th group, respectively. The precoded signal ${\bf x}$ should satisfy the power constraint $E[\|{\bf x}\|^2]\leq P$.

\subsection{Channel model}

By using the Karhunen-Loeve transform \cite{Rao_book} and the polarized MIMO channel modeling with infinitesimally small antenna elements described in \cite{KimBruno, Clerckx_book}, the downlink channel to the $g$th group, ${\bf H}_g$, can be represented as
\begin{eqnarray}\label{Sys_3}
\!{\bf H}_g \!=\! \left(\left[\begin{array}{cc}1 & r_{xp}\\r_{xp} & 1  \end{array}\right]\!\otimes \!({\bf U}_g{\bf \Lambda}_g^{\frac{1}{2}}) \right) \left({\bf G}_{g} \odot ({\bf X}\otimes {\bf 1}_{r\times \frac{\bar N}{2}} )\right)
\end{eqnarray}
where $r_{xp}$ is the correlation coefficient between vertically and horizontally polarized antenna elements, ${\bf \Lambda}_g$ is an $r_g\times r_g$ diagonal matrix with the non-zero eigenvalues of the spatial covariance matrix ${\bf R}_g^s$ for the $g$th group\footnote{For simplicity, it is assumed that the spatial covariance matrix is the same for both polarizations.}, where the eigenvalues are assumed to be ordered in decreasing order of magnitude. Note that generally, $r_g \ll M$, and ${\bf U}_g\in \mathbb{C}^{\frac{M}{2}\times r_g} $ is composed of the eigenvectors of ${\bf R}_g^s$ as its columns. The matrix ${\bf G}_{g}$ is defined as
\begin{eqnarray}\label{Sys_3_1}
{\bf G}_{g}  =\left[\begin{array}{cc}{\bf G}_{g}^{vv} & {\bf G}_{g}^{hv}\\{\bf G}_{g}^{vh} & {\bf G}_{g}^{hh}  \end{array}\right],
\end{eqnarray}
and the elements of ${\bf G}_{g}^{pq} \in \mathbb{C}^{r_g\times \frac{\bar N}{2}}$, $p,q \in \{h, v\}$ are i.i.d. complex Gaussian distributed with zero mean and unit variance. The matrix ${\bf X}$ describes the power imbalance between the orthogonal polarizations and is given as
\begin{eqnarray}\label{Sys_4}
{\bf X} = \left[\begin{array}{cc}1 & \sqrt{\chi}\\\sqrt{\chi} & 1  \end{array}\right],
\end{eqnarray}
where the parameter $0\leq {\chi} \leq 1$ is the inverse of the XPD, where $1\leq {\text{XPD}}\leq \infty$.
Note that, based on the reported measurement in \cite{Coldrey,Asplund}, $r_{xp}\approx 0$. Accordingly, (\ref{Sys_3}) can be rewritten as
\begin{eqnarray}\label{Sys_5}\nonumber
{\bf H}_g &=& \left({\bf I}_2\otimes ({\bf U}_g{\bf \Lambda}_g^{\frac{1}{2}}) \right) \left[\begin{array}{cc}{\bf G}_{g}^{vv} & \sqrt{{\chi}}{\bf G}_{g}^{hv}\\\sqrt{{\chi}}{\bf G}_{g}^{vh} & {\bf G}_{g}^{hh}  \end{array}\right]\\
&=&\left[\begin{array}{cc}{\bf U}_g{\bf \Lambda}_g^{\frac{1}{2}}{\bf G}_{g}^{vv} & \sqrt{{\chi}}{\bf U}_g{\bf \Lambda}_g^{\frac{1}{2}}{\bf G}_{g}^{hv}\\\sqrt{{\chi}}{\bf U}_g{\bf \Lambda}_g^{\frac{1}{2}}{\bf G}_{g}^{vh} & {\bf U}_g{\bf \Lambda}_g^{\frac{1}{2}}{\bf G}_{g}^{hh}  \end{array}\right] \nonumber\\&=& \left[\begin{array}{cc}{\bf H}_{g}^{vv} & {\bf H}_{g}^{hv}\\{\bf H}_{g}^{vh} & {\bf H}_{g}^{hh}  \end{array}\right],
\end{eqnarray}
and its covariance matrix is given as
\begin{eqnarray}\label{Sys_5_1}
\!{\bf R}_g &\!\!=\!\!&\left[\begin{array}{cc}(1+ \chi){\bf R}_g^s & {\bf 0}\\{\bf 0} &(1+ \chi){\bf R}_g^s \end{array}\right]\!\\\!&\!\!=\!\!&\! \left[\begin{array}{cc}{\bf R}_g^s & {\bf 0}\\{\bf 0} & \chi{\bf R}_g^s \end{array}\right] + \left[\begin{array}{cc} \chi{\bf R}_g^s & {\bf 0}\\{\bf 0} &{\bf R}_g^s \end{array}\right] = {\bf R}_{gv} +{\bf R}_{gh},\nonumber
\end{eqnarray}
where ${\bf R}_{gv}$ and ${\bf R}_{gh}$ are the covariance matrices of the vertically and the horizontally co-polarized MS subgroups, respectively.

Note that the long-term parameters ${\bf R}_g^s$ and $\chi$ are slowly-varying and assumed to be obtained accurately with a low feedback overhead. However, the short-term CSI parameter ${\bf G}_{g}$ is varying independently over the short-term coherence time. The feedback to the BS of the CSI is imperfect (due to e.g. quantization) and incurs a significant overhead. Accordingly, the imperfect CSIT $\hat{\bf G}_{g}$ available at the transmitter is modeled as
\begin{eqnarray}\label{Sys_7}
\hat{\bf G}_{g} = \sqrt{1- \tau_g}{\bf G}_{g}+ \tau_g{\bf Z}_{g},
\end{eqnarray}
where the elements of ${\bf Z}_{g}$ are complex Gaussian distributed with zero mean and unit
variance and $\tau_g \in [0,1]$ indicates the accuracy of available CSIT for the $g$th group. That is, the case of $\tau_g =0$ implies the perfect CSIT. From (\ref{Sys_5}), $\hat{\bf H}_g$ and $\hat{\bf H}_g^{pp}$ can be defined as the imperfect CSI knowledge of ${\bf H}_g$ and ${\bf H}_g^{pp}$ at the transmitter, respectively, by using $\hat{\bf G}_{g}$. Throughout the paper, it is assumed that $\tau_1=...=\tau_G = \tau$, but it can be easily extended to the scenario that $\tau_i\neq \tau_j$ for $i\neq j$.

\begin{remark}
The imperfect channel model (\ref{Sys_7}) comes from the scenario that when both BS and MS know the long-term statistics perfectly (i.e., ${\bf R}_g^s$), the $k$th MS in the $g$th group quantizes ${\bf g}_{gk} = [{\bf G}_{g}]_k$ by using the random codebook \cite{N_Jindal2} and feeds the codeword index back to the BS. Note that the $k$th MS in the $g$th group can have a much smaller feedback overhead by sending the essential channel information of ${\bf g}_{gk} \in \mathbb{C}^{2r_g \times 1}$ rather than ${\bf h}_{gk}$.
\end{remark}

\section{Dual structured precoding based on long-term/short-term CSIT}\label{sec:dualprecoding}

Thanks to the computational complexity reduction and the feedback overhead reduction (i.e., the dimension reduction using long-term statistics), the dual precoding scheme based on long-term/short-term CSIT has been widely utilized  \cite{Caire, LimYooClerckx, Clerckx_book}. That is, the precoding matrix for the $g$th group is given as
\begin{eqnarray}\label{dual_1}
{\bf V}_g = {\bf B}_g {\bf P}_g,
\end{eqnarray}
where ${\bf B}_g \in \mathbb{C}^{M\times \bar B}$ is the preprocessing matrix based on the long-term channel statistics with $\bar N \leq \bar B \leq 2 r_g \ll M$ and ${\bf P}_g \in \mathbb{C}^{\bar B\times \bar N}$ is the precoding matrix for the effective (instantaneous) channel ${\bf H}_g^H {\bf B}_g$. Here, $\bar B$ is a design parameter that determines the dimension of the transformed channel using the long-term CSIT.
The system equation (\ref{Sys_1}) can then be rewritten as
\begin{eqnarray}\label{dual_2}
{\bf y}_g =  {\bf H}_g^H  {\bf B}_g {\bf P}_g{\bf d}_g + \sum_{l=1, l\neq g}^G  {\bf H}_g^H{\bf B}_l {\bf P}_l{\bf d}_l + {\bf n}_g.
\end{eqnarray}
In what follows, we introduce the conventional dual structured precoding scheme with the preprocessing using block diagonalization (BD) based on spatial correlation and the regularized ZF precoding for each decoupled group. Then, we propose the dual precoding scheme with BD and subgrouping (BDS) exploiting both the spatial correlation and the polarization (another long-term channel statistics parameter).

\subsection{Preprocessing using block diagonalization based on spatial correlation}\label{ssec:preprocess_spatial}
To null out the leakage to other groups, it is desirable that the preprocessing matrix ${\bf B}_g$ based on the spatial correlation is designed as
\begin{eqnarray}\label{dual_3}
{\bf H}_l^H{\bf B}_g \approx {\bf 0}, \text{ for } l \neq g.
\end{eqnarray}
Then, ${\bf P}_g$ in (\ref{dual_2}) can be computed based on the decoupled system model ${\bf y}_g \approx  {\bf H}_g^H  {\bf B}_g {\bf P}_g{\bf d}_g  + {\bf n}_g$ where the inter-group interferences have been eliminated. \footnote{Note that if the BS has a large number of antenna elements and the number of antenna elements at the BS is larger than the number of MSs, we can find the ``approximated'' null-space satisfying (\ref{dual_3}), in general. Furthermore, from \cite{Caire,YinGesbert} considering the one-ring channel model (see also Section VI), if the angle-of-departures (AoDs) of the multipaths from different groups are disjointed, the spatial covariance matrices of different groups become asymptotically orthogonal to each other as the number of antenna elements increases.}

To obtain ${\bf B}_g$ satisfying the condition (\ref{dual_3}), the BD can be utilized. That is, due to the block diagonal structure in (\ref{Sys_5_1}), we first define
\begin{eqnarray}\label{dual_4}
{\bf U}_{-g} = [{\bf U}_{1}^a, ..., {\bf U}_{g-1}^a, {\bf U}_{g+1}^a,...,{\bf U}_{G}^a ] \in \mathbb{C}^{\frac{M}{2}\times \sum_{l\neq g}{r_l^a}},
\end{eqnarray}
where ${\bf U}_{g}^a = [ {\bf U}_{g}]_{1:r_g^a}$ and $r_g^a (\leq r_g)$ is a design parameter reflecting the number of {\it{dominant}} eigenvalues of ${\bf R}_g^s$. That is, if we increase $r_g^a$ close to $r_g$, the BD can find the subspace more orthogonal to the signal subspace spanned by other groups' channel (the perfect orthogonality is guaranteed when $r_g^a=r_g$), while the dimension of corresponding orthogonal subspace decreases as $\frac{M}{2} - \sum_{l\neq g}{r_l^a}$.

The matrix ${\bf U}_{-g}$ in (\ref{dual_4}) then has a singular value decomposition (SVD) as
\begin{eqnarray}\label{dual_5}
\!{\bf U}_{-\!g} \!=\! [{\bf E}_{-\!g}^{(1)}, {\bf E}_{-\!g}^{(0)}]\! \left[\!\begin{array}{cc}\!{\bf \Lambda}_{-\!g}^{(1)} \!&\!\!\\\!\!&\!{\bf \Lambda}_{-\!g}^{(0)}\! \end{array}\!\right]\!{\bf V}_{-\!g}^H, {\bf E}_{-\!g}^{(0)}\! \in \!\mathbb{C}^{\frac{M}{2}\!\times \! \frac{M}{2}\!-\!\sum_{l\!\neq \! g}{r_l^a}},\!
\end{eqnarray}
where ${\bf E}_{-g}^{(1)}$ (respectively, ${\bf E}_{-g}^{(0)}$) is the left singular vectors associated with the $\sum_{l\neq g}{r_l^a}$ dominant (respectively, $\frac{M}{2}-\sum_{l\neq g}{r_l^a}$ non-dominant) singular values ${\bf \Lambda}_{-g}^{(1)}$ (respectively, ${\bf \Lambda}_{-g}^{(0)}$). Then, because $({\bf E}_{-g}^{(0)})^H {\bf U}_{-g}=0$, by defining $\tilde{\bf H}_g= ({\bf I}_2 \otimes {\bf E}_{-g}^{(0)})^H {\bf H}_g$, $\tilde{\bf H}_g$ is orthogonal to the dominant eigen-space spanned by other groups' channel. Note that the covariance matrix of $\tilde{\bf H}_g$ is then given by
\begin{eqnarray}\label{dual_6}
\tilde{\bf R}_g = ({\bf I}_2 \otimes {\bf E}_{-g}^{(0)})^H{\bf R}_g ({\bf I}_2 \otimes {\bf E}_{-g}^{(0)}),
\end{eqnarray}
and by defining $\tilde{\bf R}_g^s = ({\bf E}_{-g}^{(0)})^H{\bf R}_g^s {\bf E}_{-g}^{(0)}$, we have its eigenvalue decomposition (EVD) as
\begin{eqnarray}\label{dual_7}
\tilde{\bf R}_g^s = {\bf F}_{g}\tilde{\bf \Lambda}_{g}{\bf F}_{g}^H,
\end{eqnarray}
where ${\bf F}_{g}$ is the eigenvectors of $\tilde{\bf R}_g^s$. Then by letting ${\bf F}_{g}^{(1)} =  [{\bf F}_{g}]_{1:\frac{\bar{B}}{2}}$, the preprocessing matrix can be given as
\begin{eqnarray}\label{dual_8}
{\bf B}_g = {\bf I}_2 \otimes {\bf B}_{g}^{s}, \quad  {\bf B}_{g}^{s} ={\bf E}_{-g}^{(0)}  {\bf F}_{g}^{(1)}.
\end{eqnarray}
Accordingly, through the preprocessing matrix ${\bf B}_g$, we can transform the transmit signal for the $g$th group into the $\bar{B}$ dimensional dominant eigen-space that is orthogonal to the subspace spanned by other groups' channel. Note that, from (\ref{dual_1}) and (\ref{dual_4}), $\bar B$ and $r_g^a$ should be chosen properly to satisfy the conditions of ${\bar N} \leq {\bar B} \leq 2(\frac{M}{2} - \sum_{l\neq g}{r_l^a})$ and ${\bar B} \leq 2 r_g$. Without loss of generality, we assume that $r_1^a =... = r_G^a= r$ with a fixed $r$ satisfying the above two constraints.

\subsection{Multi-user precoding for each decoupled group}\label{ssec:precoding_decoupled}
Because, from (\ref{dual_2}), the effective channel for the $g$th group is $\bar{\bf H}_g = {\bf B}_g^H {\bf H}_g $, the corresponding covariance matrix is given by
\begin{eqnarray}\label{dual_9}
\bar{\bf R}_g&=&  {\bf B}_g^H{\bf R}_g{\bf B}_g ={\bf B}_g^H({\bf R}_{gv} +{\bf R}_{gh}){\bf B}_g =\bar{\bf R}_{gv}+\bar{\bf R}_{gh}\nonumber\\
&=&(1+ \chi)\left[\begin{array}{cc}({\bf B}_{g}^{s})^H{\bf R}_g^s {\bf B}_{g}^{s}& {\bf 0}\\{\bf 0} &({\bf B}_{g}^{s})^H{\bf R}_g^s {\bf B}_{g}^{s} \end{array}\right].
\end{eqnarray}
Furthermore, due to the preprocessing, the interferences from other groups in (\ref{dual_2}) are almost nulled out. Accordingly, the precoding matrix ${\bf P}_g$ is designed such that the intra-group interferences are nulled out based on the short-term CSIT of the $g$th group. That is, assuming the equal power allocation, the regularized ZF precoding matrix \cite{Caire, HwangClerckx} with imperfect CSIT can be computed as
\begin{eqnarray}\label{dual_10}
{\bf P}_g = {\xi}_g \hat{\bar{{\bf K}}}_g \hat{\bar{{\bf H}}}_g,
\end{eqnarray}
where
\begin{eqnarray}\label{dual_10_1}
 \hat{\bar{{\bf K}}}_g  = \left(\hat{\bar{{\bf H}}}_g\hat{\bar{{\bf H}}}_g^H + \bar B \alpha {\bf I}_{\bar B}  \right)^{-1},\quad \hat{\bar{{\bf H}}}_g = {\bf B}_g^H \hat{{\bf H}}_g.
\end{eqnarray}
Here, $\hat{\bar{{\bf H}}}_g$ is the effective channel estimate that is available at the BS and $ \alpha $ is a regularization parameter. Throughout the paper, it is set as $\alpha= \frac{\bar N}{\bar B P}$, which is equivalent with the MMSE linear filter \cite{Tse_book}. The normalization factor ${\xi}_g$ is then given as
\begin{eqnarray}\label{dual_11}
{\xi}_g^2\! =\! \frac{\bar N \frac{P}{N}}{\frac{P}{N} tr(\hat{\bar{{\bf H}}}_g^H \hat{\bar{{\bf K}}}_g ^H{\bf B}_g^H {\bf B}_g \hat{\bar{{\bf K}}}_g \hat{\bar{{\bf H}}}_g)} \!=\! \frac{\bar N}{tr(\hat{\bar{{\bf H}}}_g^H \hat{\bar{{\bf K}}}_g ^H\hat{\bar{{\bf K}}}_g \hat{\bar{{\bf H}}}_g)},
\end{eqnarray}
where the second equality is due to the fact that ${\bf B}_g^H {\bf B}_g = {\bf I}_{\bar B}$ from (\ref{dual_8}).
Denoting $ \hat{\bar{{\bf h}}}_{gk} = [\hat{\bar{{\bf H}}}_g]_k$ as the effective channel estimate of the $k$th MS in the $g$th group, the SINR of the $k$th MS in the $g$th group with $p$ polarization is then given as shown at the top of the next page.
\begin{figure*}[!t]
\scriptsize
\begin{eqnarray}\label{dual_12}\nonumber
\gamma_{gpk}^{BD} = \frac{\frac{P}{N}{\xi}_g^2 |{\bf h}_{gk}^H{\bf B}_g  \hat{\bar{{\bf K}}}_g \hat{\bar{{\bf h}}}_{gk} |^2 }{\frac{P}{N}\sum_{j\neq k}{\xi}_g^2 |{\bf h}_{gk}^H{\bf B}_g  \hat{\bar{{\bf K}}}_g \hat{\bar{{\bf h}}}_{gj} |^2 + \frac{P}{N}\sum_{l\neq g}\sum_{j}{\xi}_l^2 |{\bf h}_{gk}^H{\bf B}_l  \hat{\bar{{\bf K}}}_l \hat{\bar{{\bf h}}}_{lj} |^2+1} ,\\= \frac{\frac{P}{N}{\xi}_g^2 |{\bf h}_{gk}^H{\bf B}_g  \hat{\bar{{\bf K}}}_g {\bf B}_g^H \hat{\bf h}_{gk} |^2 }{\frac{P}{N}\sum_{j\neq k}{\xi}_g^2 |{\bf h}_{gk}^H{\bf B}_g  \hat{\bar{{\bf K}}}_g {\bf B}_g^H \hat{{\bf h}}_{gj} |^2 + \frac{P}{N}\sum_{l\neq g}\sum_{j}{\xi}_l^2 |{\bf h}_{gk}^H{\bf B}_l  \hat{\bar{{\bf K}}}_l {\bf B}_l^H \hat{{\bf h}}_{lj} |^2+1}.
\end{eqnarray}
\hrulefill \vspace*{2pt}
\end{figure*}
Accordingly, the sum rate is given by
\begin{eqnarray}\label{dual_13}
R_{BD}= \sum_{g=1}^G \sum_{p\in\{v,h\}}\sum_{k=1}^{\bar N/2}\log_2(1+ \gamma_{gpk}^{BD}),
\end{eqnarray}
where the subscript and superscript $BD$ indicate the dual precoding with Block Diagonalization based on spatial correlation.

\subsection{Dual precoding using block diagonalization and subgrouping based on both spatial correlation and polarization}\label{ssec:dual_spatial_XPD}
In Section \ref{ssec:preprocess_spatial}, the preprocessing matrix is computed based only on spatial correlation.
However, when $\chi$ becomes small (i.e., the antenna elements can favorably discriminate the orthogonally polarized signals), the interference signals through the cross-polarized channels can be naturally nulled out. This suggests that we can make the subgroups of co-polarized MSs in each group (see the second MS group in Fig. \ref{Fig_system_block}.) and let the BS precode the signal for the co-polarized subgroup by using the short-term CSIT of the transmit antenna elements having the same polarization with the associated subgroup. That is, from (\ref{Sys_1}) and (\ref{dual_2}), the received signal for the co-polarized subgroup with $p$ polarization, for $p \in\{h, v\}$, in the $g$th group can be written as
\begin{eqnarray}\label{dual_14}
{\bf y}_g^p &\!=\!&  {\bf H}_{gp}^H {\bf B}_{gp} {\bf P}_{gp}{\bf d}_{g}^p + \sum_{\substack{q \in\{h, v\}\\ q\neq p}}{\bf H}_{gp}^H  {\bf B}_{gq} {\bf P}_{gp}{\bf d}_{g}^q \nonumber\\&\!\!&+ \sum_{l=1, l\neq g}^G \sum_{q \in\{h, v\}} {\bf H}_{gp}^H {\bf B}_{lq} {\bf P}_{lq}{\bf d}_{l}^q + {\bf n}_{g}^p,
\end{eqnarray}
where ${\bf H}_{gv}= \left[\begin{array}{c}{\bf H}_{g}^{vv} \\{\bf H}_{g}^{vh} \end{array}\right]$ and ${\bf H}_{gh}= \left[\begin{array}{c} {\bf H}_{g}^{hv}\\{\bf H}_{g}^{hh}  \end{array}\right]$ from (\ref{Sys_5}). Here, ${\bf B}_{gp}$ for $p \in\{h, v\}$ are given as
\begin{eqnarray}\label{dual_15}
{\bf B}_{gv}= \left[\begin{array}{c}{\bf B}_{g}^{s} \\{\bf 0}\end{array}\right], \quad  {\bf B}_{gh}= \left[\begin{array}{c} {\bf 0}\\{\bf B}_{g}^{s} \end{array}\right],
\end{eqnarray}
where ${\bf B}_{g}^{s}$ is given in (\ref{dual_8}). Note that, when $\chi \approx 0$, we can easily find that
\begin{eqnarray}\label{dual_16}
{\bf H}_{lp}^H{\bf B}_{gq} \approx {\bf 0}, \text{ for } p \neq q.
\end{eqnarray}
Furthermore, because ${\bf H}_{gq}$, $q\neq p$ has no influence on ${\bf P}_{gp}$, the MSs do not need to feed back the instantaneous CSI from cross polarized transmit antenna elements at BS. That is, the $k$th MS having vertical (horizontal) polarization in the $g$th group can quantize the first (last) $r$ entries of ${\bf g}_{gk}$ (see also Remark 1) and feed them back to the BS with the feedback amount reduced in half.

The precoding matrix ${\bf P}_{gp}$ is then designed such that the intra-subgroup interferences are nulled out using the co-polarized short-term CSIT. That is, letting $\hat{\bf H}_{gp}$ denote the imperfect CSI knowledge at the transmitter of ${\bf H}_{gp}$, $p \in\{h, v\}$, the regularized ZF precoding matrix with imperfect CSIT can be computed as
\begin{eqnarray}\label{dual_17}
{\bf P}_{gp} = {\xi}_{gp} \hat{\bar{{\bf K}}}_{gp} \hat{\bar{{\bf H}}}_{gp},
\end{eqnarray}
where $ \hat{\bar{{\bf K}}}_{gp}  = \left(\hat{\bar{{\bf H}}}_{gp}\hat{\bar{{\bf H}}}_{gp}^H + \frac{\bar B}{2} \alpha {\bf I}_{\frac{\bar B}{2}}  \right)^{-1}$ and $\hat{\bar{{\bf H}}}_{gp} = {\bf B}_{gp}^H \hat{{\bf H}}_{gp}=({\bf B}_{g}^s)^H \hat{\bf H}_{g}^{pp}$, the effective channel estimate that is available at the BS. The normalization factor ${\xi}_{gp}$ is then given as
\begin{eqnarray}\label{dual_18}
{\xi}_{gp}^2 = \frac{\bar N /2 }{ tr(\hat{\bar{{\bf H}}}_{gp}^H \hat{\bar{{\bf K}}}_{gp}^H  \hat{\bar{{\bf K}}}_{gp} \hat{\bar{{\bf H}}}_{gp})}.
\end{eqnarray}
Assuming equal power allocation, the SINR of the $k$th MS in the subgroup with $p$ polarization of the $g$th group is then given by
\begin{eqnarray}\label{dual_19}
\gamma_{gpk}^{BDS} = \frac{\frac{P}{N}{\xi}_{gp}^2 |{\bf h}_{gpk}^H{\bf B}_{gp}  \hat{\bar{{\bf K}}}_{gp} {\bf B}_{gp}^H \hat{\bf h}_{gpk} |^2 }{IN_{gpk}},
\end{eqnarray}
where
\begin{eqnarray}\label{dual_20}\nonumber
& IN_{pgk} = \frac{P}{N}\sum_{j\neq k}{\xi}_{gp}^2 |{\bf h}_{gpk}^H{\bf B}_{gp}  \hat{\bar{{\bf K}}}_{gp} {\bf B}_{gp}^H \hat{\bf h}_{gpj} |^2 &\\\nonumber&+\frac{P}{N}\sum_{q\neq p}\sum_j {\xi}_{gq}^2 |{\bf h}_{gpk}^H{\bf B}_{gq}  \hat{\bar{{\bf K}}}_{gq} {\bf B}_{gq}^H \hat{\bf h}_{gqj} |^2 &\\&+ \frac{P}{N}\sum_{l\neq g}\sum_q \sum_{j}{\xi}_{lq}^2 |{\bf h}_{gpk}^H{\bf B}_{lq}  \hat{\bar{{\bf K}}}_{lq} {\bf B}_{lq}^H \hat{\bf h}_{lqj} |^2+1&
\end{eqnarray}
and ${\bf h}_{gpk} = [{\bf H}_{gp}]_k$ and $\hat{\bf h}_{gpk} = [\hat{\bf H}_{gp}]_k$, respectively.
Accordingly, the sum rate is given by
\begin{eqnarray}\label{dual_21}
R_{BDS}= \sum_{g=1}^G \sum_{p\in\{v,h\}}\sum_{k=1}^{\frac{\bar N}{2}}\log_2(1+ \gamma_{gpk}^{BDS}),
\end{eqnarray}
where the subscript and superscript $BDS$ indicate the dual precoding with Block Diagonalization and Subgrouping based on both spatial correlation and polarization. Note that because, when $\chi =0$, the interference from cross-polarized groups are perfectly nulled out, we can easily find that
\begin{eqnarray}\label{dual_21_1}
\gamma_{gpk}^{BDS} = \gamma_{gpk}^{BD}.
\end{eqnarray}

\section{Asymptotic performance analysis for Dual precoding methods}\label{sec:Asymptotic}
In \cite{Wagner1}, when the number of transmit antenna elements ($M$) is large, the asymptotic SINR of the regularized ZF precoding has been analyzed in spatially correlated MISO broadcasting systems with uni-polarized antennas under the imperfect CSIT and, in \cite{Caire}, the asymptotic SINR of the dual precoding with BD has been analyzed under the perfect CSIT and the uni-polarized antenna system. In this section, based on random matrix theory results \cite{Wagner1,Hachem,HoydisDebbah}, we first derive the asymptotic SINR for two different dual precoding schemes -- dual precoding with i) BD and ii) BDS under the imperfect CSIT and dual-polarized antenna system. Note that the asymptotic inter/intra interferences are evaluated over the polarization domain as well as the spatial domain, which can encounter a more generalized channel environment with a polarization. Based on the asymptotic results, we analyze the performance as a function of the XPD parameter $\chi$, and propose a new dual precoding/feedback scheme in the next section.

%based on random matrix theory results \cite{Wagner1,Hachem,HoydisDebbah}, we first derive the asymptotic SINR for two different dual precoding schemes -- dual precoding with i) BD and ii) BDS under the imperfect CSIT and dual-polarized antenna system. Note that the asymptotic inter/intra interferences are evaluated over the polarization domain as well as the spatial domain, which can encounter a more generalized channel environment with a polarization. That is, the asymptotic performance can be analyzed in terms of both the polarization and the spatial correlation Thanks to the structure of the dual polarized antenna elements, we can have a simple asymptotic SINR in Theorem \ref{thm_BD} compared to that for the case of the general antenna covariance matrices \cite{Wagner1}. This gives a useful insight into the behavior of the asymptotic SINR as a function of the polarization parameter in Section \ref{ssec:dualprecoding}

\subsection{Dual precoding with block diagonalization based on spatial correlation}\label{ssec:Asymp_BD}
Before we proceed with the derivation of the asymptotic SINR for the dual precoding with BD, we introduce an important theorem about the asymptotic behavior of a random matrix with a large dimension developed in \cite{Wagner1}.
\begin{thm} (\cite{Wagner1}, Theorem 1)\label{thm_BC} Let ${\bf H}$ be the $M\times N$ matrix, in which each column is a zero-mean complex Gaussian random vector having a covariance matrix ${\bf R}_i$ for $i=1,...,N$. In addition, let ${\bf S}, {\bf Q} \in \mathbb{C}^{M\times M}$ be Hermitian nonnegative definite. Assume $\lim \sup_{M \rightarrow \infty } \sup_{1\leq i\leq N} \| {\bf R }_i\| <\infty$ and ${\bf Q}$ has uniformly bounded spectrum norm. Then, for $z< 0$,
\begin{eqnarray}\label{Asym_1}
\frac{1}{M}tr({\bf Q}({\bf H}{\bf H}^H +{\bf S} - z {\bf I}_M)^{-1} ) - \frac{1}{M}tr({\bf Q} {\bf T}(z)) \overset{M \to \infty}{{\longrightarrow}} 0,
\end{eqnarray}
where
\begin{eqnarray}\label{Asym_2}
{\bf T}(z) = \left(\frac{1}{M}\sum_{j=1}^N\frac{{\bf R}_j}{1+e_{j}(z)} +{\bf S} -z{\bf I}_M  \right)^{-1}.
\end{eqnarray}
Here, $e_{i}(z)$ for $i=1,...,N$ are the unique solution of
\begin{eqnarray}\label{Asym_3}
e_{i}(z) =\frac{1}{M} tr\left({\bf R}_i \left(\frac{1}{M}\sum_{j=1}^N\frac{{\bf R}_j}{1+e_{j}(z)} +{\bf S} -z{\bf I}_M  \right)^{-1} \right),
\end{eqnarray}
which can be solved by the fixed-point algorithm and its convergence is also proved in \cite{Wagner1}.
\end{thm}
Then, by using Theorem \ref{thm_BC}, the asymptotic SINR for the dual precoding with BD can be derived.

\begin{thm}\label{thm_BD} When $M$, $N$, $\bar B$ goes to infinity and $\frac{N}{\bar B}$ is fixed, the SINR, $\gamma_{gpk}^{BD}$ in (\ref{dual_12}) asymptotically converges as
\begin{eqnarray}\label{Asym_4}
\gamma_{gpk}^{BD}- \gamma_{gpk}^{BD,o}~ \overset{M \to \infty}{{\longrightarrow}} 0,
\end{eqnarray}
where $\gamma_{gpk}^{BD,o}$ is the asymptotic SINR, given as shown at the top of the page.
\begin{figure*}[!t]
\scriptsize
\begin{eqnarray}\label{Asym_5}
\gamma_{gpk}^{BD,o}= \frac{\frac{P}{N}({\xi}_g^o)^2 (1-\tau^2) (m_{gp}^o)^2}{ ({\xi}_g^o)^2 \Upsilon_{ggp}^o (1-\tau^2(1-(1+ m_{gp}^o)^2))+(1+\sum_{l \neq g}({\xi}_l^o)^2\Upsilon_{glp}^o )(1+ m_{gp}^o)^2},
\end{eqnarray}
with $({\xi}_g^o)^2 = \frac{P}{G \Psi_g^o}$.

\hrulefill \vspace*{2pt}
\end{figure*}
Here, $m_{gp}^o$, $\Upsilon_{glq}^o$, and $\Psi_g^o$ are the unique solutions of
\begin{eqnarray}\label{Asym_6_1}
\!\!&\!\!m_{gp}^o\! =\! \frac{1}{\bar B}tr(\bar{\bf R}_{gp} {\bf T}_g),{\bf T}_g\! = \!\left(\frac{\bar N}{2\bar B}\!\sum_{q\in\{h,v\}}\!\frac{\bar{\bf R}_{gq}}{1+m_{gq}^o} +\alpha{\bf I}_{\bar B}\! \right)^{-\!1}\!,\!&\!\nonumber\\\!&\! \Psi_g^o = \frac{1}{2\bar B}\frac{P}{G} \sum_{q\in\{h,v\}} \frac{m'_{gq}}{(1+ m_{gq}^o)^2},\!&\\\label{Asym_6_2}
\!&\Upsilon_{ggp}^o= \frac{\bar N/2 -1}{\bar B}\frac{P}{N}\frac{m'_{ggpp}}{(1+m_{gp}^o)^2} +\frac{\bar N}{2\bar B}\frac{P}{N}\frac{m'_{ggpq}}{(1+m_{gq}^o)^2}, \!&\!\nonumber\\ \!&\!\Upsilon_{glp}^o= \frac{P}{2N}\frac{\bar N}{\bar B}\sum_{q\in\{h,v\}} \frac{m'_{glpq}}{(1+m_{lq}^o)^2}.\!&\!
\end{eqnarray}
In addition, ${\bf m}'_g = [m'_{gv}, m'_{gh}]^T$, and ${\bf m}'_{ggp} = [m'_{ggpv}, m'_{ggph}]^T$ given by
\begin{eqnarray}\label{Asym_6_3}
{\bf m}'_{g} = ({\bf I}_2 - {\bf J})^{-1}{\bf v}_g, \quad
{\bf m}'_{ggp} = ({\bf I}_2 - {\bf J})^{-1}{\bf v}_{ggp},
\end{eqnarray}
where
\begin{eqnarray}\label{Asym_6_4}
&\!\!{\bf J}= \frac{\bar N}{2\bar B}\left[\begin{array}{cc} \frac{ tr(\bar{\bf R}_{gv}{\bf T}_g \bar{\bf R}_{gv} {\bf T}_g) }{\bar B(1+m_{gv}^o)^2} & \frac{ tr(\bar{\bf R}_{gv}{\bf T}_g \bar{\bf R}_{gh} {\bf T}_g) }{\bar B(1+m_{gh}^o)^2}\\  \frac{ tr(\bar{\bf R}_{gh}{\bf T}_g \bar{\bf R}_{gv} {\bf T}_g) }{\bar B(1+m_{gv}^o)^2} & \frac{ tr(\bar{\bf R}_{gh}{\bf T}_g \bar{\bf R}_{gh} {\bf T}_g) }{\bar B(1+m_{gh}^o)^2}\end{array}\right],\!&\!\\\nonumber\!&\!{\bf v}_g\!=\! \frac{1}{\bar B}\left[\!\begin{array}{c} tr(\bar{\bf R}_{gv}{\bf T}_g^2)\\ tr(\bar{\bf R}_{gh}{\bf T}_g^2)\end{array}\!\right],~{\bf v}_{ggp}\!= \!\frac{1}{\bar B}\left[\!\begin{array}{c} tr(\bar{\bf R}_{gv}{\bf T}_g\bar{\bf R}_{gp}{\bf T}_g)\\ tr(\bar{\bf R}_{gh}{\bf T}_g\bar{\bf R}_{gp}{\bf T}_g) \end{array}\!\right]\!&\!
\end{eqnarray}
with $\bar{\bf R}_{gp}$ defined in (\ref{dual_9}). In addition,
${\bf m}'_{glp} = [m'_{glpv}, m'_{glph}]^T$ given by
\begin{eqnarray}\label{Asym_6_5}
{\bf m}'_{glp} = ({\bf I}_2 - {\bf J})^{-1}{\bf v}_{glp},
\end{eqnarray}
where
\begin{eqnarray}\label{Asym_6_6}
{\bf J}= \frac{\bar N}{2\bar B}\left[\begin{array}{cc} \frac{ tr(\bar{\bf R}_{lv}{\bf T}_l \bar{\bf R}_{lv} {\bf T}_l) }{\bar B(1+m_{lv}^o)^2} & \frac{ tr(\bar{\bf R}_{lv}{\bf T}_l \bar{\bf R}_{lh} {\bf T}_l) }{\bar B(1+m_{lh}^o)^2}\\  \frac{ tr(\bar{\bf R}_{lh}{\bf T}_l \bar{\bf R}_{lv} {\bf T}_l) }{\bar B(1+m_{lv}^o)^2} & \frac{ tr(\bar{\bf R}_{lh}{\bf T}_l \bar{\bf R}_{lh} {\bf T}_l) }{\bar B(1+m_{lh}^o)^2}\end{array}\right],\nonumber\\~{\bf v}_{glq}= \frac{1}{\bar B}\left[\begin{array}{c} tr(\bar{\bf R}_{lv}{\bf T}_l {\bf B}_l^H{\bf R}_{gp}{\bf B}_l{\bf T}_l)\\ tr(\bar{\bf R}_{lh}{\bf T}_l{\bf B}_l^H{\bf R}_{gp}{\bf B}_l{\bf T}_l) \end{array}\right].
\end{eqnarray}
\end{thm}
\begin{proof}
See Appendix A.
\end{proof}
Thanks to the structure of the dual polarized antenna elements, we can have a simple asymptotic SINR in Theorem \ref{thm_BD} compared to that for the case of the general antenna covariance matrices \cite{Wagner1}. This gives a useful insight into the behavior of the asymptotic SINR as a function of the polarization parameter in Section \ref{ssec:dualprecoding}. Furthermore, when the spatial covariance matrix is the same for both polarizations, we can further simplify the asymptotic SINR in Theorem \ref{thm_BD}.
\begin{cor}\label{cor_BD}
When the spatial covariance matrix is the same for both polarization, i.e., infinitesimally small dual-polarized antenna elements are co-located (see footnote 2), the asymptotic SINR $\gamma_{gpk}^{BD,o}$ in (\ref{Asym_4}) can be written as in a simpler form as shown at the top of the page.
\begin{figure*}[!t]
\scriptsize
\begin{eqnarray}\label{Asym_8}
\gamma_{gk}'^{BD,o}= \frac{\frac{P}{N}({\xi}_g^o)^2 (1-\tau^2) (m_g^o)^2}{ ({\xi}_g^o)^2 \Upsilon_{gg}^o (1-\tau^2(1-(1+ m_g^o)^2))+(1+\sum_{l \neq g} ({\xi}_l^o)^2 \Upsilon_{gl}^o )(1+ m_g^o)^2},
\end{eqnarray}
with $({\xi}_g^o)^2 = \frac{P}{G \Psi_g^o}$.

\hrulefill \vspace*{2pt}
\end{figure*}
Here, $m_g^o$, $\Upsilon_{gl}^o$, and $\Psi_g^o$ are the unique solutions of
\begin{eqnarray}\label{Asym_9_1}
m_g^o = \frac{1}{\bar B}tr(\bar{\bf R}_g' {\bf T}_g),\quad {\bf T}_g = \left(\frac{\bar N}{\bar B}\frac{\bar{\bf R}_g'}{1+m_g^o} +\alpha{\bf I}_{\bar B} \right)^{-1},
\\\label{Asym_9_2}
\Psi_g^o = \frac{1}{\bar B}\frac{P}{G}\frac{m'_g}{(1+ m_g^o)^2},\quad m'_g= \frac{\frac{1}{\bar B}tr(\bar{\bf R}_g'{\bf T}_g^2) }{1-\frac{\frac{\bar N}{\bar B} tr(\bar{\bf R}_g'{\bf T}_g \bar{\bf R}_g' {\bf T}_g) }{\bar B(1+m_g^o)^2} },\\\label{Asym_9_3}
\Upsilon_{gg}^o= \frac{\bar N -1}{\bar B}\frac{P}{N}\frac{m'_{gg}}{(1+m_g^o)^2}, \quad \Upsilon_{gl}^o= \frac{P}{N}\frac{\bar N}{\bar B}\frac{m'_{gl}}{(1+m_l^o)^2},\\\label{Asym_9_4}
\!\!\!m'_{gg} \!=\! \frac{\frac{1}{\bar B}tr(\bar{\bf R}_g'{\bf T}_g\bar{\bf R}_g'{\bf T}_g) }{\!1-\frac{\frac{\bar N}{\bar B} tr(\bar{\bf R}_g'{\bf T}_g \bar{\bf R}_g' {\bf T}_g) }{\bar B(1+m_g^o)^2}\! },~\! m'_{gl}\! = \!\frac{\!\frac{1}{\bar B}tr(\bar{\bf R}_l'{\bf T}_l{\bf B}_l^H{\bf R}_g'{\bf B}_l{\bf T}_l) \!}{1-\frac{\frac{\bar N}{\bar B} tr(\bar{\bf R}_l'{\bf T}_l \bar{\bf R}_l' {\bf T}_l) }{\bar B(1+m_l^o)^2} },\!\!
\end{eqnarray}
where ${\bf R}'_{g} = \frac{1}{2}{\bf R}_g$ and $\bar{\bf R}_g' = \frac{1}{2}\bar{\bf R}_g$. Here, ${\bf R}_g$ and $\bar{\bf R}_g$ are defined in (\ref{Sys_5_1}) and (\ref{dual_9}).
\end{cor}
\begin{proof}
From (\ref{Sys_5_1}) and (\ref{dual_9}), we can see that $m_{gh}^o  = m_{gv}^o = \frac{1}{\bar B}tr(\frac{1}{2}\bar{\bf R}_{g} {\bf T}_g)$ in (\ref{Asym_6_1}). Accordingly, by letting $m_{g}^o = \frac{1}{\bar B}tr(\frac{1}{2}\bar{\bf R}_{g} {\bf T}_g)$, ${\bf T}_g$ in (\ref{Asym_6_1}) can be rewritten as that in (\ref{Asym_9_1}). Furthermore, because, in (\ref{Asym_6_4}),
\begin{eqnarray}\label{Asym_10}
\!\!&\!\!tr(\bar{\bf R}_{gv}{\bf T}_g^2)= tr(\bar{\bf R}_{gh}{\bf T}_g^2)= \frac{1}{2}tr(\bar{\bf R}_{g}{\bf T}_g^2),\!&\!\!\\\nonumber
\!\!&\!tr(\bar{\bf R}_{gv}{\bf T}_g \bar{\bf R}_{gv} {\bf T}_g) \!+\! tr(\bar{\bf R}_{gv}{\bf T}_g \bar{\bf R}_{gh} {\bf T}_g) \!= \! \frac{1}{4}tr(\bar{\bf R}_{g}{\bf T}_g \bar{\bf R}_{g} {\bf T}_g),\!&\!\!
\end{eqnarray}
$m'_{gv}= m'_{gh} = m'_{g}$ as in (\ref{Asym_9_2}). Similarly, we can prove that the parameters $m'_{glpq}$ are given as (\ref{Asym_9_4}). By substituting $m'_{g}$, $m'_{gg}$, $m'_{gl}$ into $\Psi_g^o$, $\Upsilon_{ggp}^o$, and $\Upsilon_{glp}^o$ of (\ref{Asym_6_1}) and (\ref{Asym_6_2}), we can prove that $\Psi_g^o$, $\Upsilon_{ggp}^o$, and $\Upsilon_{glp}^o$ can be written as in (\ref{Asym_9_2}) and (\ref{Asym_9_3}).
\end{proof}
From Corollary \ref{cor_BD}, the asymptotic SINR is independent of MS index $k$ and polarization index $p$. Accordingly, by letting $\gamma_{gk}'^{BD,o}\triangleq \gamma_{g}'^{BD,o}$ for $k=1,...,\bar N$, the asymptotic sum rate can be approximated as
\begin{eqnarray}\label{Asym_10_1}
R_{BD}^o &\approx& \sum_{g=1}^G \sum_{p\in\{v,h\}}\sum_{k=1}^{\bar N/2}\log_2(1+ \gamma_{gk}'^{BD,o})\nonumber\\&=& \sum_{g=1}^G \bar N \log_2(1+ \gamma_{g}'^{BD,o}).
\end{eqnarray}

\begin{remark}\label{remark2}
We note that, when $\tau = 0$, $\gamma_{gk}'^{BD,o}$ in Corollary \ref{cor_BD} is analogous to the asymptotic SINR of the dual precoding with BD derived in \cite{Caire} under the perfect CSIT and uni-polarized system. That is, when the infinitesimally small dual-polarized antenna elements are co-located, the asymptotic SINRs of the MSs in the same group are the same irrespective of their antenna deployment, i.e., vertical or horizontal polarization. In addition, the effective covariance matrix is given by ${\bf R}'_{g} = \frac{1}{2}{\bf R}_g$. That is, it can be described as if the BS and MSs are co-polarized and the correlation matrices for the MSs in the $g$th group are the same as ${\bf I}_2\otimes {\bf R}_g^s$ and the effective transmit power of BS is reduced from ${P}$ to $\frac{1+\chi}{2}{P}$. Note that this is valid only when the spatial covariance matrix is the same for both polarizations. That is, Theorem \ref{thm_BD} is extended to more general covariance matrices addressing the polarization of antenna elements.
\end{remark}

\subsection{Dual precoding with  block diagonalization and subgrouping based on both spatial correlation and polarization}\label{ssec:Asymp_BDS}
By using Theorem \ref{thm_BC} and an approach similar as that used for the dual precoding with BD, the asymptotic SINR for the dual precoding with BDS can be derived.

\begin{thm}\label{thm_BDS} When $M$, $N$, $\bar B$ go to infinity and $\frac{N}{\bar B}$ is fixed, the SINR, $\gamma_{gpk}^{BDS}$ in (\ref{dual_19}) asymptotically converges as
\begin{eqnarray}\label{Asym_11}
\gamma_{gpk}^{BDS}- \gamma_{gpk}^{BDS,o}~ \overset{M \to \infty}{{\longrightarrow}} 0,
\end{eqnarray}
where $\gamma_{gpk}^{BDS,o}$ is given by
\begin{eqnarray}\label{Asym_12}
\gamma_{gpk}^{BDS,o} = \frac{\frac{P}{N}({\xi}_{gp}^o)^2 (1-\tau^2) (m_{gp}^o)^2}{IN_{gpk}^{o}},
\end{eqnarray}
where
\begin{eqnarray}\label{Asym_12_1}
IN_{gpk}^{o} =  ({\xi}_{gp}^o)^2 \Upsilon_{ggpp}^o (1-\tau^2(1-(1+ m_{gp}^o)^2))+\nonumber\\(1+\sum_{q \neq p}({\xi}_{gq}^o)^2\Upsilon_{ggpq}^o  +\sum_{l \neq g}\sum_{q}({\xi}_{lq}^o)^2\Upsilon_{glpq}^o )(1+ m_{gp}^o)^2,
\end{eqnarray}
with $({\xi}_{gp}^o)^2 = \frac{P}{G \Psi_{gp}^o}$, where $m_{gp}^o$, $\Upsilon_{glpq}^o$, and $\Psi_{gp}^o$ are the unique solutions of
\begin{eqnarray}\label{Asym_13_1}
\!&\!m_{gp}^o\! =\! \frac{2}{\bar B}tr(\bar{\bf R}_{gp} {\bf T}_{gp}),~ {\bf T}_{gp} \!=\! \left(\frac{\bar N}{\bar B}\frac{\bar{\bf R}_{gp}}{1+m_{gp}^o} +\alpha{\bf I}_{\bar B/2} \right)^{-1},\nonumber\!&\!\\ \!&\!\Psi_{gp}^o = \frac{1}{\bar B}\frac{P}{G} \frac{m'_{gp}}{(1+ m_{gp}^o)^2},\!&\!\\\label{Asym_13_2}
\!\!&\!\!\Upsilon_{ggpp}^o\!=\! \frac{\bar N/2 -1}{\bar B/2}\frac{P}{N}\frac{m'_{ggpp}}{(1+m_{gp}^o)^2}, ~ \Upsilon_{glpq}^o\!=\! \frac{P}{N}\frac{\bar N}{\bar B} \frac{m'_{glpq}}{(1+m_{lq}^o)^2},\!\!&\!
\end{eqnarray}
with $m'_{gp} \!= \!\frac{\frac{2}{\bar B}tr(\bar{\bf R}_{gp}{\bf T}_{gp}^2) }{1-\frac{\frac{\bar N}{\bar B} tr(\bar{\bf R}_{gp}{\bf T}_{gp} \bar{\bf R}_{gp} {\bf T}_{gp}) }{\bar B/2(1+m_{gp}^o)^2} }$,
\begin{eqnarray}\label{Asym_14_1}
\! m'_{ggpp} \!=\! \frac{\frac{2}{\bar B}tr(\bar{\bf R}_{gp}{\bf T}_{gp}\bar{\bf R}_{gp}{\bf T}_{gp}) }{1-\frac{\frac{\bar N}{\bar B} tr(\bar{\bf R}_{gp}{\bf T}_{gp} \bar{\bf R}_{gp} {\bf T}_{gp}) }{\bar B/2(1+m_{gp}^o)^2} }\!\\\!
m'_{glpq} = \frac{\frac{2}{\bar B}tr(\bar{\bf R}_{lq}{\bf T}_{lq}{\bf B}_{lq}^H {\bf R}_{gp}{\bf B}_{lq} {\bf T}_{lq}) }{1-\frac{\frac{\bar N}{\bar B} tr(\bar{\bf R}_{lq}{\bf T}_{lq} \bar{\bf R}_{lq}{\bf T}_{lq}) }{\bar B/2(1+m_{lq}^o)^2} },
\end{eqnarray}
where $\bar{\bf R}_{gp}$ is defined in (\ref{dual_9}).
\end{thm}
\begin{proof}
Because the proof is similar to that of Theorem \ref{thm_BD}, it is omitted.
\end{proof}
From Theorem \ref{thm_BDS}, the asymptotic SINR is independent of MS index $k$. Furthermore, because ${\bf X}$ in (\ref{Sys_4}) is symmetric, $\gamma_{gvk}^{BDS,o}= \gamma_{ghk}^{BDS,o}$. Accordingly, by letting $\gamma_{gvk}^{BDS,o}= \gamma_{ghk}^{BDS,o} \triangleq \gamma_{g}^{BDS,o}$, the asymptotic sum rate can be approximated as
\begin{eqnarray}\label{Asym_15}
R_{BDS}^o &\approx& \sum_{g=1}^G \sum_{p\in\{v,h\}}\frac{\bar N}{2}\log_2(1+ \gamma_{gpk}^{BDS,o})\nonumber\\&=&\sum_{g=1}^G {\bar N}\log_2(1+ \gamma_{g}^{BDS,o}).
\end{eqnarray}

\subsection{Asymptotic performance analysis as a function of the XPD parameter $\chi$}\label{ssec:dualprecoding}
In this section, we investigate the effect of the XPD parameter $\chi$ on the asymptotic SINRs of the dual precoding schemes.
%The following lemma is helpful to derive its effect analytically.
%\begin{lem} Let a symmetric nonnegative matrix ${\bf R}$ have the EVD as ${\bf R}= {\bf U}{\bf \Lambda}{\bf U}^H$, where ${\bf U} = [{\bf U}_1 {\bf U}_0]$ and ${\bf \Lambda} = diag\{{\bf \Lambda}_1, {\bf 0} \}$ with a diagonal matrix ${\bf \Lambda}_1$ having non-zero eigenvalues of ${\bf R}$. By letting ${\bf T} = \left(\beta {\bf R} + \alpha {\bf I} \right)^{-1}$, for small $\alpha$, ${\bf R}{\bf T}$ can be approximated as
%\begin{eqnarray}\label{Asym_15_1}
%{\bf R}{\bf T} &\approx& \frac{1}{\beta}{\bf R}\left({\bf R} + \alpha {\bf I} \right)^{-1}.
%\end{eqnarray}
%\end{lem}
Based on the asymptotic results in Theorems \ref{thm_BD} and \ref{thm_BDS} and Corollary \ref{cor_BD}, we can have the following propositions.

\begin{prop}\label{prop_BD_Xpd} For a large $M$, the asymptotic SINR of the dual precoding with BD in (\ref{Asym_5}) is approximately independent of the XPD parameter $\chi$. That is, if we define the asymptotic SINR as the function $\gamma_{gpk}^{BD,o}(\chi)$ of $\chi$, then $\gamma_{gpk}^{BD,o}(\chi)\approx \gamma_{gpk}^{BD,o}(0)$.
\end{prop}
\begin{proof}
See Appendix B.
\end{proof}

\begin{prop}\label{prop_BDS_Xpd} For a large $M$, the asymptotic SINR of the dual precoding with BDS in (\ref{Asym_12}) can be approximately written as the function $\gamma_{gpk}^{BDS,o}(\chi)$ of $\chi$ given by
\begin{eqnarray}\label{Asym_22}
\gamma_{gpk}^{BDS,o}(\chi)\approx \frac{\gamma_{gpk}^{BDS,o}(0)}{ 1+ c_0 \chi},
\end{eqnarray}
where
\begin{eqnarray}\label{Asym_22_1}
\!\!c_0\! =\! E_{g, p}\left[\!\frac{({\xi}_{gp}^o(0))^2\Upsilon_{ggpp}^o(0)}{\frac{({\xi}_{gp}^o(0))^2\Upsilon_{ggpp}^o(0)}{(1+ m_{gp}^o(0))^2}(\tau^2((1\!+\! m_{gp}^o(0))^2\!-\!1)\! +\!1)  \!+\!1}\!\right].\!\!
\end{eqnarray}
\end{prop}
\begin{proof}
See Appendix C.
\end{proof}

\begin{remark}\label{remark3}
Note that, from Proposition 2, the asymptotic SINR of the dual precoding with BDS decreases when $\chi$ increases. This is because the subgroups are formed with the assumption that the interferences through the cross-polarized channels are perfectly nulled out in Section \ref{ssec:dual_spatial_XPD} and ${\bf P}_{gp}$ in (\ref{dual_17}) is determined based only on co-polarized CSIT. Therefore, the interference power increases proportionally to $\chi$. In contrast, because the dual precoding with BD nulls out the intra-group interferences based on both co/cross polarized CSIT, it exhibits performances somehow robust to the variation of the polarization parameter $\chi$. In addition, from Theorem \ref{thm_BD}, \ref{thm_BDS}, and Corollary \ref{cor_BD}, we can also see that $\gamma_{gpk}^{BDS,o}(0) = \gamma_{gpk}^{BD,o}(0)$.
\end{remark}
%\begin{remark}\label{remark4}
%Because $\frac{({\xi}_{gp}^o(0))^2\Upsilon_{ggpp}^o(0)}{ (1+ m_{gp}^o(0))^2}$ implies an asymptotic intra-subgroup interference in (\ref{Asym_12_1}), for low SNR (i.e., the noise-limited system, $\frac{({\xi}_{gp}^o(0))^2\Upsilon_{ggpp}^o(0)}{ (1+ m_{gp}^o(0))^2} \ll 1$ ),
%\begin{eqnarray}\label{Asym_22_2}
%c_{0,l} \approx  \frac{1}{2G}\sum_{g=1}^G \sum_{p\in\{v,h\}} {({\xi}_{gp}^o(0))^2\Upsilon_{ggpp}^o(0)},
%\end{eqnarray}
%which is independent of the channel accuracy parameter $\tau$. For high SNR (i.e., the interference-limited system, $\frac{({\xi}_{gp}^o(0))^2\Upsilon_{ggpp}^o(0)}{ (1+ m_{gp}^o(0))^2} \gg 1$ and $m_{gp}^o(0)\gg1$),
%\begin{eqnarray}\label{Asym_22_3}
%c_{0,h} \approx \frac{1}{2G}\sum_{g=1}^G \sum_{p\in\{v,h\}} \frac{(m_{gp}^o(0))^2}{\tau^2(m_{gp}^o(0))^2 +1},
%\end{eqnarray}
%which is inversely proportional to the channel accuracy parameter $\tau$. That is, at high SNR, as the channel is more accurate, the performance of BDS is more affected by $XPD$.
%%Accordingly, $c_0$ can be approximated as
%%\begin{eqnarray}\label{Asym_22_3}
%%c_{0} &=& E_{g, p}\left[\frac{1}{\frac{1}{(1+ m_{gp}^o(0))^2}(\tau^2((1+ m_{gp}^o(0))^2-1) +1)  +\frac{1}{({\xi}_{gp}^o(0))^2\Upsilon_{ggpp}^o(0)}}\right]\nonumber\\ &\approx &\frac{1}{ E_{g, p}\left[\frac{1}{(m_{gp}^o(0))^2}(\tau^2( m_{gp}^o(0))^2 +1)\right]  +E_{g, p}\left[\frac{1}{({\xi}_{gp}^o(0))^2\Upsilon_{ggpp}^o(0)}\right]}  \approx \frac{1}{1/c_{0,h} + 1/c_{0,l} }.
%%\end{eqnarray}
%\end{remark}

\section{Discussion}
\label{sec:newprecoding}
\subsection{A new dual structured precoding/feedback}\label{ssec:newprecoding}

Even though the sum-rate performance of the dual precoding with BDS decreases as $\chi$ increases, it can utilize more accurate short-term CSIT compared to the dual precoding with BD under the same number of feedback bits as stated in Section \ref{ssec:dual_spatial_XPD}.
Assuming the CSI is perfectly estimated at MSs, when random vector quantization (RVQ) with $N_B$ bits is utilized \cite{N_Jindal2}, the quantization error for the short-term CSIT in the dual precoding with BD (i.e., the columns of ${\bf G}_{g}$ in (\ref{Sys_3_1})) is upper bounded as\footnote{Here, the codebook is fixed given $r$ and there is no adaptive codebook that would adapt as a function of $r$ and XPD. Of course, if $r$ changes, the codebook can be adaptively designed with respect to $r$, but it is out of scope of this paper.}
\begin{eqnarray}\label{newprecoding_1}
\tau_{BD}^2 < 2^{-\frac{N_B}{2r-1}}.
\end{eqnarray}
For the dual precoding with BDS, the quantization error is upper bounded as
\begin{eqnarray}\label{newprecoding_2}
\tau_{BDS}^2 < 2^{-\frac{N_B}{r-1}}.
\end{eqnarray}
Because the bound is tight for a large $N_B$ \cite{N_Jindal2}, by assuming $\tau_{BD}^2 = 2^{-\frac{N_B}{2r-1}} (\approx \tau_{BDS}$), we have the following proposition.

\begin{prop}\label{prop_Swhitching} For a given $\chi$ and a large $M$, when
%\begin{eqnarray}\label{newprecoding_3}
%B  \leq ((1+ m_g^o(0))^2 + \frac{(1+ m_g^o(0))^2 -1}{ ({\xi}_g^o(0))^2 \Upsilon_{gg}^o(0)}) 2^{-\frac{B}{2r-1}},
%\end{eqnarray}
\begin{eqnarray}\label{newprecoding_3}
\!\!\!N_B \!\!&\!\!\lesssim\! \!&\!\!(2r\!-\!1)\!\Biggl(\!\log_2\!\!\left(\!1\!+\!E_{g,p}\!\!\left(\! \frac{(1\!+\! m_{gp}^o(0))^2 -1}{ ({\xi}_{gp}^o(0))^2\! \Upsilon_{ggpp}^o(0)(1\!+ \! m_{gp}^o(0))^2} \!\right)\!\!\right)\!\! \nonumber\\ \!\! &\!\!\!\!&\!- \log_2 \chi \!\Biggr),\!\!\!
\end{eqnarray}
the dual precoding with BDS outperforms that with BD.
\end{prop}
\begin{proof}
From (\ref{Asym_12}) and Proposition 2, the SINR of the dual precoding with BDS can be written as
\begin{eqnarray}\label{newprecoding_4}
\!\!\!\gamma_{gpk}^{BDS,o}(\chi)\!=\! \frac{A_0(1-\tau_{BDS}^2)}{\!( B_0 (1 \!+\!D_0\tau_{BDS}^2)\!+\! (1\!+\!E_0)(D_0\!+\!1)) (1\!+\!c_0\chi)\!},\!\!\!
\end{eqnarray}
where $A_0 = \frac{P}{N}({\xi}_{gp}^o(0))^2(m_{gp}^o(0))^2$, $B_0 = ({\xi}_{gp}^o(0))^2 \Upsilon_{ggpp}^o(0) $, $D_0 = (1+ m_{gp}^o(0))^2 -1 $, and $E_0= \sum_{l \neq g} ({\xi}_{lp}^o(0))^2 \Upsilon_{glpp}^o(0)$. Note that $\gamma_{gpk}^{BD,o}(0) = \gamma_{gpk}^{BDS,o}(0)$ in Remark \ref{remark3}, assuming the same channel accuracy (i.e., $\tau_{BD}^2 = \tau_{BDS}^2$). Therefore, by setting $\chi=0$ in (\ref{newprecoding_4}), from Proposition 1, the SINR of the dual precoding with BD can then be given as
\begin{eqnarray}\label{newprecoding_5}
\gamma_{gpk}^{BD,o}(\chi)= \frac{A_0(1-\tau_{BD}^2)}{ B_0 (1+D_0\tau_{BD}^2)+ (1+E_0)(D_0+1) }.
\end{eqnarray}
Because the dual precoding with BDS outperforms that with BD when $\gamma_{gpk}^{BDS,o}(\chi) \geq \gamma_{gpk}^{BD,o}(\chi) $, by letting $\tau_{BD}^2 \triangleq \tau^2 = 2^{-\frac{N_B}{2r-1}}(\approx \tau_{BDS})$, we have
\begin{eqnarray}\label{newprecoding_6}
\frac{(1-\tau^4)}{( B_0 (1+D_0\tau^4)+ (1+E_0)(D_0+1)  )(1+c_0\chi)} \nonumber\\\geq \frac{(1-\tau^2)}{ B_0 (1 +D_0\tau^2)+ (1+E_0)(D_0+1)  }.
\end{eqnarray}
After a simple calculation, we have
\begin{eqnarray}\label{newprecoding_7}
\chi \leq \frac{(B_0D_0 +B_0 +(1+E_0)(D_0+1) )\tau^2 }{c_0 (B_0D_0 \tau^4+B_0 +(1+E_0)(D_0+1) )}.
\end{eqnarray}
From (\ref{Asym_22_1}) and the fact that ${E}_0 \ll 1$ due to the BD, $c_0 \approx \frac{B_0(D_0+1)}{B_0 D_0\tau^4 +B_0 +D_0 +1}$ and we have
\begin{eqnarray}\label{newprecoding_8}
\!\!\chi &\!\leq\!& (1+ \frac{D_0}{B_0(D_0+1)})\tau^2 \nonumber\\\!&\!=\!&\!\left(1+ \frac{(1+ m_{gp}^o(0))^2 -1}{ ({\xi}_{gp}^o(0))^2 \Upsilon_{ggpp}^o(0)(1+ m_{gp}^o(0))^2} \right)\tau^2,\!
\end{eqnarray}
which induces (\ref{newprecoding_3}) by taking the expectation over $g$ and $p$ in (\ref{newprecoding_8}).
\end{proof}

\begin{remark}\label{remark4}
From Proposition \ref{prop_Swhitching}, when the feedback bits are not enough to describe the short-term CSIT accurately, the dual precoding with BDS exhibits a better performance than that with BD. That is, it is preferable that by forming the co-polarized subgroup, each MS feeds back the short-term CSI from the co-polarized transmit antenna elements.
In addition, from (\ref{newprecoding_3}), when $\chi \rightarrow 0 $, the dual precoding with BDS always exhibits better performance than that with BD. Note that for high SNR (i.e., $m_{gp}^o(0)\gg 1$), the expectation term in (\ref{newprecoding_3}) is approximated as $E_{g,p}\left( \frac{1}{ ({\xi}_{gp}^o(0))^2 \Upsilon_{ggpp}^o(0)} \right)$, which is inversely proportional to the intra-subgroup interference power from (\ref{Asym_12_1}). Hence, a smaller intra-subgroup interference widens the region where BDS outperforms BD. That is, if the transmit signals to the co-polarized MSs can be asymptotically well separated by the linear precoding (less intra-subgroup interference), the feedback of the short-term CSI of the co-polarized channel with a higher accuracy is preferable.
\end{remark}
Therefore, motivated by Proposition \ref{prop_Swhitching}, a new dual precoding/feedback scheme can be described in Fig \ref{Fig_proposed}. Note that, depending on the long-term CSI (spatial correlation, polarization) and the number of feedback bits (or, the short-term CSIT accuracy $\tau$), the dual precoding is switched between BD and BDS. In other words, by analyzing the asymptotic performance of BD and BDS, we can propose a new dual structured precoding which outperforms both BD and BDS by balancing the weakness of BD (low CSIT accuracy across the whole array) and BDS (performance degradation due to large polarization parameter $\chi$).
That is, given the same feedback overhead, for small $\chi$, the dual precoding with BDS will exhibit better performance, while, for large $\chi$, that with BD will show better performance. However, because the new dual precoding is switched between BD and BDS based on Proposition 3, the new dual precoding will show better performance than other two schemes.
\begin{figure}
\begin{center}
\begin{tabular}{c}
\includegraphics[height=4.2cm]{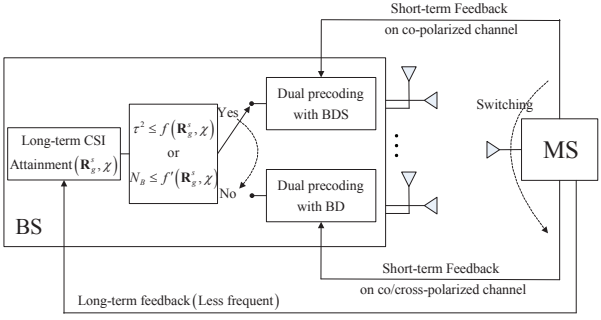}
\end{tabular}
\end{center}
\caption[Fig_proposed]
%>>>> use \label inside caption to get Fig. number with \ref{}
{ \label{Fig_proposed} Block diagram of a new dual precoding/feedback scheme.}
\end{figure}

\subsection{3D dual structured precoding}\label{ssec:3D_BF}
Motivated by 3D beamforming \cite{Caire,LiJi}, the proposed scheme can also be extended to the scenario of 3D dual structured precoding. Assuming that the ${M_E}\times{M_A}$ uniform planar array with dual-polarized antenna elements is exploited at BS and there are $L$ elevation regions. For simplicity, each elevation region has the same number of groups, $G$, as in Fig \ref{Fig_3D_BF}. Note that the elevation angular spread depends on the distance $d_{gl}$ between the BS and the $g$th group of the $l$th region and the radius of the ring of scatterer $s_{gl}$. By letting ${\bf h}_{g_kl}$ be the $2M_AM_E\times 1$ vectorized channel of the $k$th MS in the $g$th group of the $l$th region, it can be written as
\begin{eqnarray}\label{3DBF_1}
{\bf h}_{g_kl} = (({\bf I}_2\otimes{\bf U}_{glA})\otimes {\bf U}_{lE})({\bf \Lambda}_{glA}^{\frac{1}{2}}\otimes {\bf \Lambda}_{lE}^{\frac{1}{2}}){{\bf g}_{g_kl}},
\end{eqnarray}
where ${\bf \Lambda}_{glA}$ and ${\bf \Lambda}_{lE}$ are the $r_{glA}\times r_{glA}$ and $r_{lE}\times r_{lE}$ diagonal matrices with non-zero eigenvalues of the spatial correlation matrices
${\bf R}_{glA}^s$ and ${\bf R}_{lE}^s$ over the azimuth and elevation directions, respectively, and ${\bf U}_{glA}$ and ${\bf U}_{lE}$ are the matrices of the associated eigenvectors. Here, ${\bf g}_{g_kl} = [ {\bf g}_{g_klv}^T,~ \chi {\bf g}_{g_klh}^T]^T$ (resp. ${\bf g}_{g_kl} = [ \chi{\bf g}_{g_klv}^T,~ {\bf g}_{g_klh}^T]^T$) for vertically (resp. horizontally) polarized MSs and ${\bf g}_{g_klp}$ is a $r_{glA} r_{lE}\times 1$ vector whose elements are complex Gaussian distributed with zero mean and unit variance. Then, the 3D dual structured precoding signal can be given as
\begin{eqnarray}\label{3DBF_2}
{\bf x} = \sum_{l=1}^L(\sum_{g=1}^G{\bf V}_{gl} {\bf d}_{gl} )\otimes {\bf q}_l,
\end{eqnarray}
where ${\bf q}_l$ is the preprocessing vector based on ${\bf R}_{lE}^s$ that nulls out the interferences from the other elevation regions. Similarly to Section \ref{ssec:preprocess_spatial}, ${\bf q}_l$ can be computed such that ${\bf q}_l^H {\bf U}_{-lE}=0 $ with ${\bf U}_{-lE}= [{\bf U}_{1E},...,{\bf U}_{l-1E},{\bf U}_{l+1E},...,{\bf U}_{LE}]$. Then, after a simple manipulation, the received signal ${\bf y}_{gl}$ of the $g$th group in the $l$th region is given as
\begin{eqnarray}\label{3DBF_3}
\!\!\!{\bf y}_{gl}\!&\!\!\!\! =\!\!\!\!&\!  \sum_{g'\!=\!1}^{G}\!{\bf G}_{gl}^H (\!{\bf \Lambda}_{glA}^{\frac{1}{2}}\!\otimes \!{\bf \Lambda}_{lE}^{\frac{1}{2}}\!)\!(\!(\!({\bf I}_2\!\otimes\!{\bf U}_{glA} \!)^H\!{\bf V}_{g'l} {\bf d}_{g'l})\!\otimes\! (\!{\bf U}_{lE}^H{\bf q}_l\!)\!) \!+\!{\bf n}_{gl}\!\!\nonumber\\
\!\!& \!\!\!\!\!= \!\!\!\!&\! \sum_{g'=1}^G\sqrt{\tilde\lambda_{l}}{\bf G}_{gl}^H{\bf \Lambda}_{glA}^{\frac{1}{2}}({\bf I}_2\otimes{\bf U}_{glA} )^H{\bf V}_{g'l} {\bf d}_{g'l} +{\bf n}_{gl},\!\!
\end{eqnarray}
where $\tilde\lambda_{l} = {\bf q}_l^H{\bf R}_{lE}^s{\bf q}_l$. Note that after the vertical preprocessing based on long-term CSIT (elevation), (\ref{3DBF_3}) is the (2-D spatial domain) equivalent system in Section \ref{sec:systemmodel} and the new dual structured precoding in Section \ref{ssec:newprecoding} can be applied to (\ref{3DBF_3}) with a channel scaling constant $\tilde\lambda_{l}$ for the $l$th elevation region.
\begin{figure}
\begin{center}
\begin{tabular}{c}
\includegraphics[height=4.5cm]{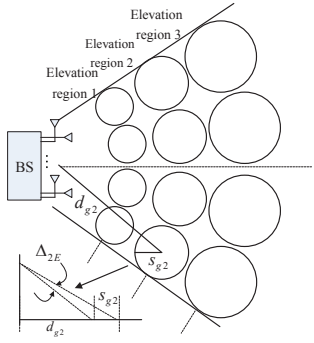}
\end{tabular}
\end{center}
\caption[Fig_3D_BF]
%>>>> use \label inside caption to get Fig. number with \ref{}
{ \label{Fig_3D_BF} Block diagram of a 3D dual polarized multi-user downlink system.}
\end{figure}

\subsection{Polarization mismatch}\label{ssec:polmismatch}
When the antenna polarization between the transmitter and the receiver is not perfectly aligned (i.e., polarization mismatch), the performance can be degraded for both single polarized antenna system and dual polarized antenna system \cite{HJoungYook,Coldrey}. In \cite{HJoungYook}, the polarization mismatch can be expressed by the rotation matrix with a mismatched angle $\theta^{ms}$ (i.e. $\left[\begin{smallmatrix}\cos\theta^{ms}& -\sin\theta^{ms}\\\sin\theta^{ms}&\cos\theta^{ms}\end{smallmatrix}\right]$). Accordingly, we let $\theta_{gk}^{ms}\sim U[-\theta_{\max}^{ms}, \theta_{\max}^{ms}]$ denote the mismatched polarization angle for the $k$th user in the $g$th group, where $U[a, b]$ indicates the uniform distribution between $a$ and $b$ and $\theta_{\max}^{ms}$ is the maximum value of the mismatched angle. Note that the case of $\theta_{\max}^{ms}=0$ indicates the MSs are perfectly aligned with either vertically or horizontally polarized antenna elements of the BS. Then, from (\ref{Sys_5}) and the above observation, the polarization mismatched channel ${\bf h}_{gk}^{ms}$ for the $k$th MS in the $g$th group can be given as
\begin{eqnarray}\label{Sys_5_1_rev}
{\bf h}_{gk}^{ms} =\left[\begin{array}{c}{\bf U}_g{\bf \Lambda}_g^{\frac{1}{2}} (\cos \theta_{gk}^{ms} - \sqrt{{\chi}} \sin\theta_{gk}^{ms} )({\bf g}_{gk})_{1:r_g}\\{\bf U}_g{\bf \Lambda}_g^{\frac{1}{2}}(\sin \theta_{gk}^{ms} + \sqrt{{\chi}} \cos\theta_{gk}^{ms} )({\bf g}_{gk})_{r_g+1:2r_g}\end{array}\right] ,
\end{eqnarray}
for vertical polarized MSs, and
\begin{eqnarray}
{\bf h}_{gk}^{ms} =\left[\begin{array}{c}{\bf U}_g{\bf \Lambda}_g^{\frac{1}{2}} (\sqrt{{\chi}}\cos \theta_{gk}^{ms} -  \sin\theta_{gk}^{ms} )({\bf g}_{gk})_{1:r_g}\\{\bf U}_g{\bf \Lambda}_g^{\frac{1}{2}}(\sqrt{{\chi}}\sin \theta_{gk}^{ms} +  \cos\theta_{gk}^{ms} )({\bf g}_{gk})_{r_g+1:2r_g}\end{array}\right],\label{Sys_5_2_rev}
\end{eqnarray}
for horizontal polarized MSs. From (\ref{Sys_5_1_rev}), the covariance matrix of the vertically co-polarized MS subgroup ${\bf R}_{gv}^{ms}$ can be given as
\begin{eqnarray}\label{Sys_5_3_rev}
{\bf R}_{gv}^{ms} =\left[\begin{array}{cc}{\bf U}_g{\bf \Lambda}_g^{\frac{1}{2}} {\bf D}_{vv}^{ms} {\bf \Lambda}_g^{\frac{H}{2}} {\bf U}_g^H&{\bf 0}\\ {\bf 0} & {\bf U}_g{\bf \Lambda}_g^{\frac{1}{2}} {\bf D}_{hv}^{ms} {\bf \Lambda}_g^{\frac{H}{2}} {\bf U}_g^H\end{array}\right],
\end{eqnarray}
where
\begin{eqnarray}\label{Sys_5_4_rev}
{\bf D}_{vv}^{ms} =E\left[ \frac{2}{\bar N}\sum_{k=1}^{\frac{\bar N}{2}}(\cos \theta_{gk}^{ms} - \sqrt{{\chi}} \sin\theta_{gk}^{ms} )^2\right] {\bf I}_{r_g} \nonumber\\=\left(\frac{1}{2} +\frac{\sin2\theta_{\max}^{ms}}{4\theta_{\max}^{ms}}+\chi \left(\frac{1}{2} -\frac{\sin2\theta_{\max}^{ms}}{4\theta_{\max}^{ms}}\right) \right) {\bf I}_{r_g},
\nonumber\\
{\bf D}_{hv}^{ms} = E\left[ \frac{2}{\bar N}\sum_{k=1}^{\frac{\bar N}{2}}(\sin \theta_{gk}^{ms} + \sqrt{{\chi}} \cos\theta_{gk}^{ms} )^2\right] {\bf I}_{r_g} \nonumber\\=\left(\frac{1}{2} -\frac{\sin2\theta_{\max}^{ms}}{4\theta_{\max}^{ms}}+\chi \left(\frac{1}{2} +\frac{\sin2\theta_{\max}^{ms}}{4\theta_{\max}^{ms}}\right) \right) {\bf I}_{r_g}.\nonumber
\end{eqnarray}
Therefore, we can have
\begin{eqnarray}\label{Sys_5_5_rev}
\!{\bf R}_{gv}^{ms} \!=\! c_{eff}\!\left[\!\begin{array}{cc}{\bf R}_g^s & {\bf 0}\\{\bf 0} & \chi_{eff}\!{\bf R}_g^s \end{array}\!\right], ~ {\bf R}_{gh}^{ms}\! =\! c_{eff}\!\left[\!\begin{array}{cc}\chi_{eff}\!{\bf R}_g^s & {\bf 0}\\{\bf 0} & {\bf R}_g^s \end{array}\!\right]\!
\end{eqnarray}
where $c_{eff} = \left(\frac{1}{2} +\frac{\sin2\theta_{\max}^{ms}}{4\theta_{\max}^{ms}}+\chi \left(\frac{1}{2} -\frac{\sin2\theta_{\max}^{ms}}{4\theta_{\max}^{ms}}\right) \right)  $ and $\chi_{eff}= \frac{1}{c_{eff}} \left(\frac{1}{2} -\frac{\sin2\theta_{\max}^{ms}}{4\theta_{\max}^{ms}}+\chi \left(\frac{1}{2} +\frac{\sin2\theta_{\max}^{ms}}{4\theta_{\max}^{ms}}\right) \right)$. Note that $c_{eff} \leq 1$ and the equality is satisfied when $\theta_{\max}^{ms} =0$, i.e., polarization is perfectly aligned. Interestingly, it is known that the effect of the polarization mismatch can be captured by the multiplication of the long-term channel with a scalar ($\leq 1$) %such as the channel path loss
\cite{Coldrey}, which is also consistent with $c_{eff}$ in (\ref{Sys_5_5_rev}). Furthermore, considering the polarization mismatch, we can use $\chi_{eff}$ instead of $\chi$ in the dual precoding with the mode switching described in Remark \ref{remark4} and Fig. \ref{Fig_proposed}.

\section{Simulation Results}
\label{sec:simulation}

Throughout the simulations, we consider the one-ring model for the spatially correlated channel \cite{Shiu_book, ShiuFoschini}. The correlation between the channel coefficients of antenna elements $1\leq m, n \leq M$ is given by
\begin{eqnarray}\label{Sys_6}
[{\bf R}_g^s ]_{m,n} = \frac{1}{2\Delta_g}\int_{-\Delta_g}^{\Delta_g}e^{-j \pi \lambda_{0}^{-1} {\bf \Omega}(\alpha +\theta_g)({\bf r}_m - {\bf r}_n)}d\alpha,
\end{eqnarray}
where $\theta_g$ and $\Delta_g$ are, respectively, the azimuth angle at which the $g$th group is located and the angular spread of the departure waves to the $g$th group which is determined as $\Delta_g \approx tan^{-1} (s_g/d_g)$. Here, $s_g$ and $d_g$ are, respectively, the radius of the ring of scatterers for the $g$th group and the distance between the BS and the $g$th group. (see Fig. \ref{Fig_system_block}). The parameter $ \lambda_{0}$ is the wavelength of signal, ${\bf r}_m =[x_m, y_m]^T$ is the position vector of the $m$th antenna element, and ${\bf \Omega}(\alpha)$ indicates the wave vector with the angle-of-departure (AoD), $\alpha$, given by ${\bf \Omega}(\alpha) = (\cos(\alpha), \sin(\alpha))$.

\subsubsection{Suitability of multi-polarized antenna elements in multi-user Massive MIMO system}
To see the suitability of the dual polarized antenna elements in multi-user Massive MIMO system, we have compared the performance of the linear array antennas with dual polarized antenna elements and single polarized antenna elements, when the dual structured precoding scheme with BD is applied. The number of antenna elements in both single polarized and dual polarized linear arrays are set as $M=120$. The dual polarized linear array is then composed of 60 vertically and 60 horizontally polarized antenna elements. Here, the XPD parameter is set as $\chi = 0.1$. Total 32 MSs with a single antenna element are clustered into 4 groups having the same number of MSs per group ($N=32$, $G=4$, $\bar N = 8$). Here, we assumed that when the single polarized linear array is deployed at the BS, the antenna elements of all MSs are co-polarized with those of BS (which is the optimal scenario for the single polarized linear array). For the dual polarized array case, in each group, $\frac{\bar N}{2}$ vertically polarized and $\frac{\bar N}{2}$ horizontally polarized MSs coexist.\footnote{Note that this environment can be made by choosing proper $\frac{\bar N}{2}$ vertically polarized and $\frac{\bar N}{2}$ horizontally polarized MSs when we have enough MSs in a cell. user selection (multi-user diversity) in this paper.} For preprocessing, we set $\bar B = 14$ for both the single/dual polarized cases such that $\bar N \leq \bar B \leq ({M} - (G-1)r')$ and $\bar B \leq r'$, where $r'$ is the minimum among the rank of ${\bf R}_g$, $g=1,...,G$.. In addition, $\Delta_1=...=\Delta_4 = \Delta = \frac{8}{180}\pi$ and $\theta_g = -\frac{\pi}{4} +\frac{\pi}{6}(g-1)$ for $g=1,...,4$.
Fig. \ref{Fig3_Dual_vs_Single} shows the sum-rate curves for the arrays with dual polarized antenna elements, $d_s = \frac{\lambda_0}{2}$ and with single polarized antenna elements, $d_s = \{\frac{\lambda_0}{2}, \frac{\lambda_0}{4}\}$, where $d_s$ is the inter-antenna distance.
We can see that the dual polarized array with the array size of $30\lambda_0$ exhibits better performance than the single polarized one with the size of $60\lambda_0$ (corresponding to $d_s = \frac{\lambda_0}{2}$). Furthermore, we can see that, if we let the single polarized array have the same size as the dual polarized one by reducing its inter-antenna space, its performance worsens significantly. Accordingly, the multi-polarized antenna can be one possible solution that partially alleviates the space limitation of Massive MIMO system.
\begin{figure}
\begin{center}
\begin{tabular}{c}
\includegraphics[height=5cm]{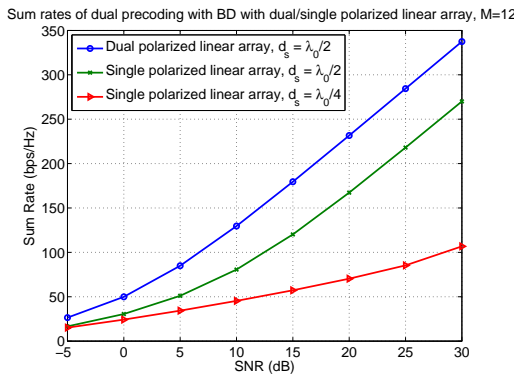}
\end{tabular}
\end{center}
\caption[Fig3_Dual_vs_Single]
%>>>> use \label inside caption to get Fig. number with \ref{}
{ \label{Fig3_Dual_vs_Single} Sum rates of the dual precoding with BD with dual/single polarized linear arrays.}
\end{figure}

\subsubsection{Performance comparison of dual precodings with BD and BDS}

To evaluate the performance of the dual precoding schemes with BD and BDS, we have also run Monte-Carlo (MC) simulations. Here, it is assumed that BS has a dual polarized linear array antenna with $M=120$. It is also assumed that $N=32$, $G=4$, $\bar N = 8$. For preprocessing, we set as $\bar B = \min ( 2 \bar N, 2r )$, where $r$ is the minimum among the rank of ${\bf R}_g^s$, $g=1,...,G$. In addition, ${\bf R}_g^s$ is generated by (\ref{Sys_6}) with $d_s = \frac{\lambda_0}{2}$, $\Delta = \frac{\pi}{12}$, and $\theta_g = -\frac{\pi}{4} +\frac{\pi}{6}(g-1)$ for $g=1,...,4$.

Fig. \ref{Fig4_perfectCSIT} shows the sum rate of the dual precoding with BD and BDS when the perfect CSIT is assumed (i.e., $\tau^2=0$). Note that when $\chi =0$, the dual precoding with BD and BDS exhibit the same sum rate performance, as mentioned in Remark \ref{remark3} and (\ref{dual_21_1}). However, when $\chi =0.1$, the performance of the dual precoding with BDS is degraded, while that of the dual precoding with BD does not change significantly compared to the case of $\chi =0$. Fig. \ref{Fig5_imperfectCSIT} shows the sum rates when $\chi=0$ for the imperfect CSIT with $\tau^2 =0.1$ (i.e., $\tau_{BD}^2= \tau^2$ and $\tau_{BDS}^2= \tau$ from (\ref{newprecoding_1}) and (\ref{newprecoding_2})). We can see that the BDS exhibits better performance than the BD because the BDS can exploit more accurate short-term CSIT due to the feedback of the smaller dimensional channel instance. Interestingly, the gap between the analytic results and MC simulation results becomes larger as SNR increases. This is because the analytical derivation is based on Theorem \ref{thm_BC} with $z = -\alpha$ and the approximation error bound is proportional to $\frac{1}{\alpha} =\frac{\bar B P}{\bar N}$, which is also addressed in [Proposition 12, \citen{Wagner1}]. In Fig. \ref{Fig6_XpdperfectCSIT}, we provide the sum rate curves of dual precoding schemes as a function of $\chi$ for $\tau^2=\{0, 0.5, 1.0, 1.5\}$ when $SNR = 15dB$. Here, the approximated sum-rate is obtained from the approximated SINR in Proposition \ref{prop_BD_Xpd} and \ref{prop_BDS_Xpd}. Note that the performance of the dual precoding with BD is not affected by the variation of $\chi$, while that with BDS decreases as $\chi$ increases. However, the dual precoding with BD is largely affected by $\tau^2$. As $\tau^2$ increases, the region that BDS outperforms BD becomes wider. Note that the crossing point in Fig. \ref{Fig6_XpdperfectCSIT} is located where $\chi \approx \tau^2$, which agrees with (\ref{newprecoding_8}) when $B_0 \gg 1$.

\begin{figure}
\begin{center}
\begin{tabular}{c}
\includegraphics[height=5cm]{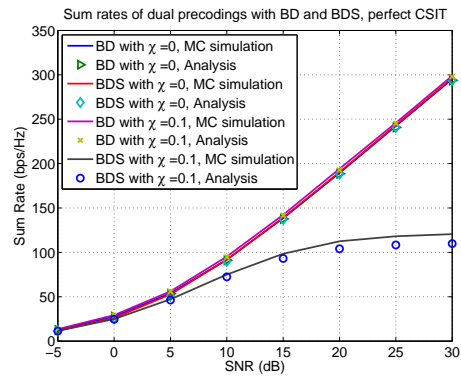}
\end{tabular}
\end{center}
\caption[Fig4_perfectCSIT]
%>>>> use \label inside caption to get Fig. number with \ref{}
{ \label{Fig4_perfectCSIT} Sum rates of the dual precoding with BD and BDS under the perfect CSIT.}
\end{figure}

\begin{figure}
\begin{center}
\begin{tabular}{c}
\includegraphics[height=5cm]{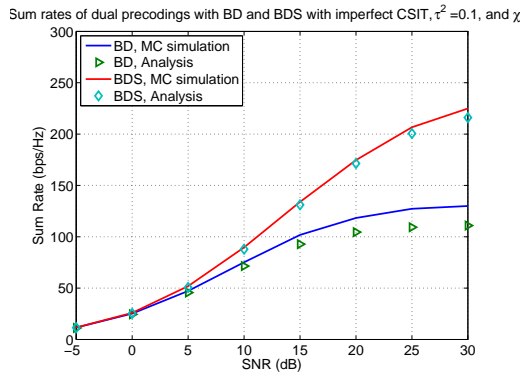}
\end{tabular}
\end{center}
\caption[Fig5_imperfectCSIT]
%>>>> use \label inside caption to get Fig. number with \ref{}
{ \label{Fig5_imperfectCSIT} Sum rates of the dual precoding with BD and BDS under the imperfect CSIT, $\tau^2 = 0.1$.}
\end{figure}

\begin{figure}%[htbp]
\centering %\hspace{-3em}
 \subfigure[]
  {\includegraphics[height=4.1cm]{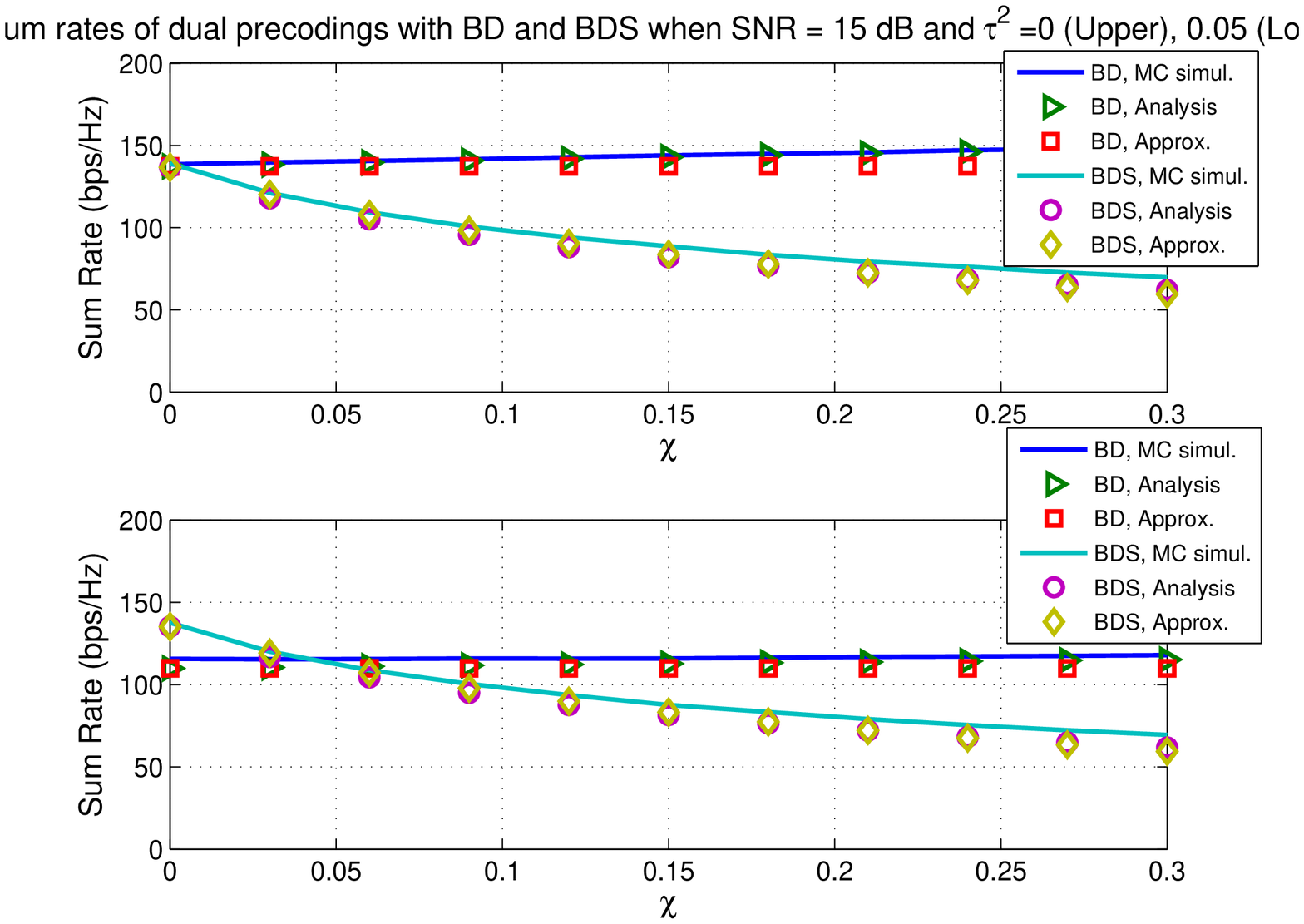}}
 \subfigure[]
  {\includegraphics[height=4.1cm]{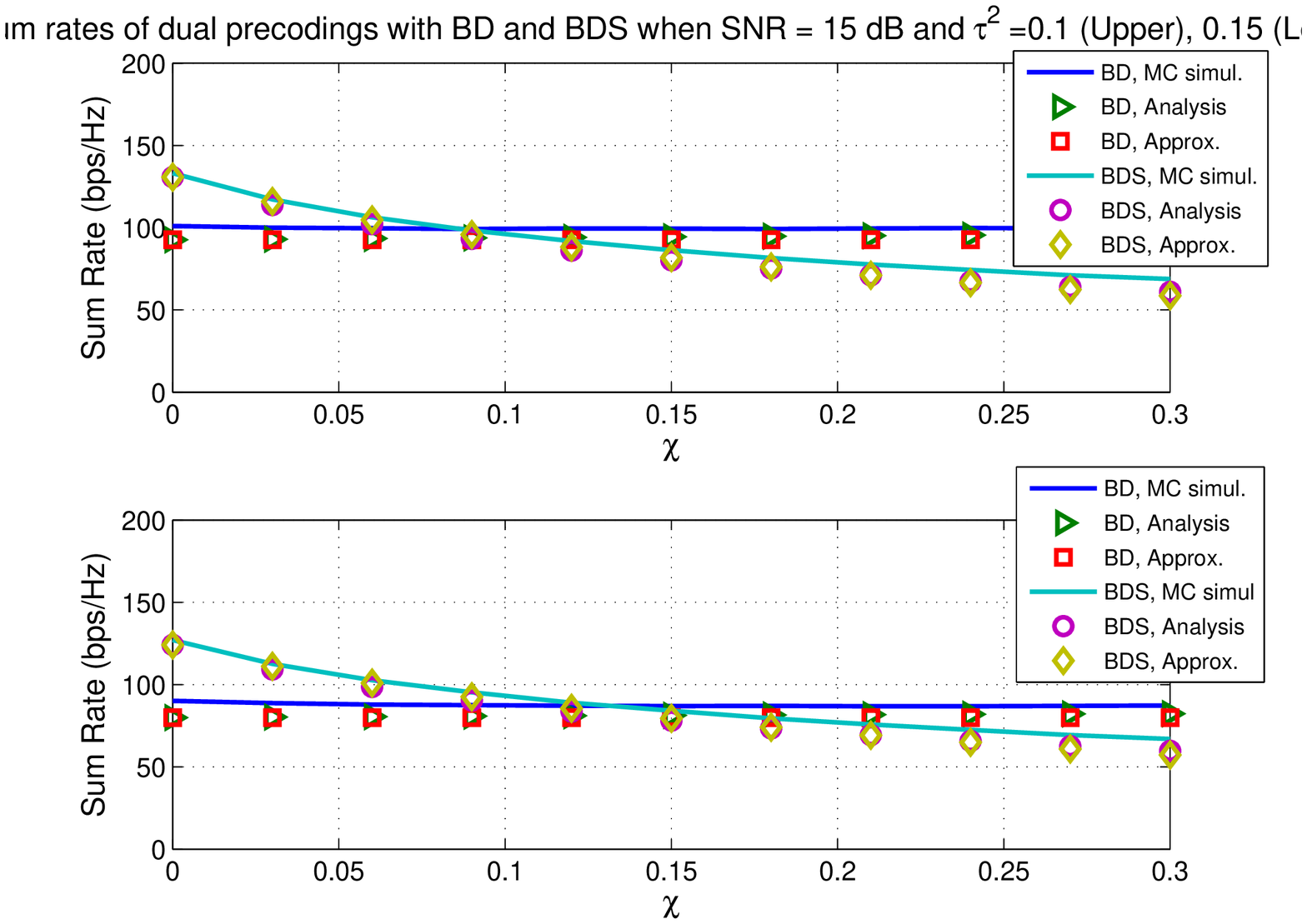}}
 \caption{\label{Fig6_XpdperfectCSIT} Sum rates of the dual precoding with BD and BDS over $\chi$ for (a) $\tau^2=\{0, 0.5\}$ (b) $\tau^2= \{ 1.0, 1.5\}$ when $SNR=15dB$.}
\end{figure}

In Fig. \ref{Fig8_N_bit}, we also compare the sum rates of BD and BDS versus the number of feedback bits per user when $\chi =\{0.1, 0.2\}$, $SNR=25$. Per Remark \ref{remark4}, when the feedback bits are not enough to describe the short-term CSIT accurately, the BDS exhibits a better performance than the BD. Here, the cross point corresponds to $N_B \approx (2r-1)\left(\log_2\left(1+E_{g,p}\left( \frac{(1+ m_{gp}^o(0))^2 -1}{ ({\xi}_{gp}^o(0))^2 \Upsilon_{ggpp}^o(0)(1+ m_{gp}^o(0))^2} \right)\right) - \log_2 \chi \right)$ in (62).

\begin{figure}
\begin{center}
\begin{tabular}{c}
\includegraphics[height=5cm]{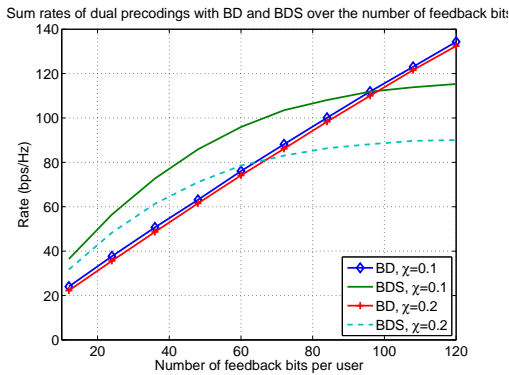}
\end{tabular}
\end{center}
\caption[Fig8_N_bit]
%>>>> use \label inside caption to get Fig. number with \ref{}
{ \label{Fig8_N_bit} Sum rates of the dual precoding with BD and BDS versus the number of feedback bits for $\chi =\{0.1, 0.2\}$.}
\end{figure}

\subsubsection{Performance comparison of the proposed dual precoding}

%To verify the performance of the proposed dual structured precoding in Section \ref{sec:newprecoding}, we have compared its sum rates with those of dual precodings with BD and BDS. In Fig. \ref{Fig9_proposed_SNR}, we set $N_B = \{50, 65\}$ for BD and BDS and $N_B = \{49, 64\}$ for the proposed dual precoding because the proposed dual precoding requires one additional bit for the information of the selected mode between BD and BDS. Here, $\chi$ is uniformly distributed on $[0,0.5]$. We can find that the sum rates for $N_B=65$ is higher than that for $N_B=50$ and the sum rates of all schemes are saturated at high SNR due to the imperfect CSIT. Note that the proposed scheme exhibits higher performance than the other two schemes, as mentioned in Section V.A.

To verify the performance of the proposed dual structured precoding in Section \ref{sec:newprecoding}, we have compared its sum rates with those of dual precodings with BD and BDS. In Fig. \ref{Fig9_proposed_SNR}, we set $N_B = \{50, 65\}$ and $\chi$ is uniformly distributed on $[0,0.5]$. We can find that the sum rates for $N_B=65$ is higher than that for $N_B=50$ and the sum rates of all schemes are saturated at high SNR due to the imperfect CSIT. Note that the proposed scheme exhibits higher performance than the other two schemes, as mentioned in Section V.A.

In Fig. \ref{Fig11_3Dproposed_XPD}, we have evaluated the 3D dual structured precoding in Section \ref{ssec:3D_BF} when $10\times 50$ uniform planar array is deployed at BS with a height of $60m$ and there are three elevation regions with $d_{gl}\in\{ 30, 60, 100\}m$, each with $4$ groups. We assume the equal power allocation over the elevation region, $SNR= 25dB$, and that $\tau^2$ is uniformly distributed on $[0,1]$. Each group has $8$ MSs and all groups have the same angular spread $\Delta = \frac{\pi}{12}$. The radius of ring of scatterer is then given as $s_{gl}=d_{gl}tan(\Delta)$ and the path loss is modeled as $P_{loss,l}= \frac{1}{1+\left(\frac{d_{gl}}{60}\right)^3} $. Then, the elevation angle spread can be computed as $tan^{-1}\left(\frac{60}{d_{gl}-s_{gl}}\right)-tan^{-1}\left(\frac{60}{d_{gl}}\right)$ (see Fig. \ref{Fig_3D_BF}). In addition, the mismatch angle for each user is randomly generated as $\theta_{gk}^{ms}\sim U[-\theta_{\max}^{ms}, \theta_{\max}^{ms}]$ with (a) $\theta_{\max}^{ms} = 0$  and (b) $\theta_{\max}^{ms} = 0.22\pi$.
We can see that the overall performances for $\theta_{\max}^{ms}=0.22{\pi}$ is outperformed by those for $\theta_{\max}^{ms}=0$ (i.e., the polarization is perfectly aligned). In addition, the performance of dual precoding with BDS is more sensitively degraded than that with the BD. For small $\chi$ the dual precoding with BDS exhibits better performance than that with BD regardless of polarization mismatch. Furthermore, the proposed dual precoding in Section \ref{sec:newprecoding} also outperforms the two other schemes. That is, by switching between BD and BDS based on the long-term CSIT parameters (spatial correlation and XPD) jointly with the number of short-term feedback bits (or, short-term CSIT quality), the performance of the dual structured precoding can be improved, regardless of the polarization mismatch. In addition, using the effective XPD parameter $\chi_{eff}$ shows better performance than the actual XPD parameter $\chi$.

\begin{figure}
\begin{center}
\begin{tabular}{c}
\includegraphics[height=4.8cm]{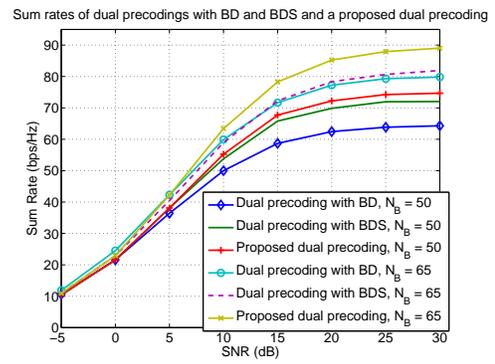}
\end{tabular}
\end{center}
\caption[Fig9_proposed_SNR]
%>>>> use \label inside caption to get Fig. number with \ref{}
{ \label{Fig9_proposed_SNR} Sum rates of dual precodings with BD and BDS and a proposed dual precoding with $N_B = \{50, 65\}$.}
\end{figure}

%\begin{figure}
%\begin{center}
%\begin{tabular}{c}
%\includegraphics[height=5cm]{Fig10_proposed_XPD}
%\end{tabular}
%\end{center}
%\caption[Fig10_proposed_XPD]
%%>>>> use \label inside caption to get Fig. number with \ref{}
%{ \label{Fig10_proposed_XPD} Sum rates of dual precodings with BD and BDS and a proposed dual precoding over $\chi$.}
%\end{figure}
%%
%
%\begin{figure}
%\begin{center}
%\begin{tabular}{c}
%\includegraphics[height=5cm]{Fig11_3Dproposed_XPD}
%\end{tabular}
%\end{center}
%\caption[Fig11_3Dproposed_XPD]
%%>>>> use \label inside caption to get Fig. number with \ref{}
%{ \label{Fig11_3Dproposed_XPD} Sum rates of the 3D dual precodings over $\chi$.}
%\end{figure}
%

\begin{figure}%[htbp]
\centering %\hspace{-3em}
 \subfigure[]
  {\includegraphics[height=4.1cm]{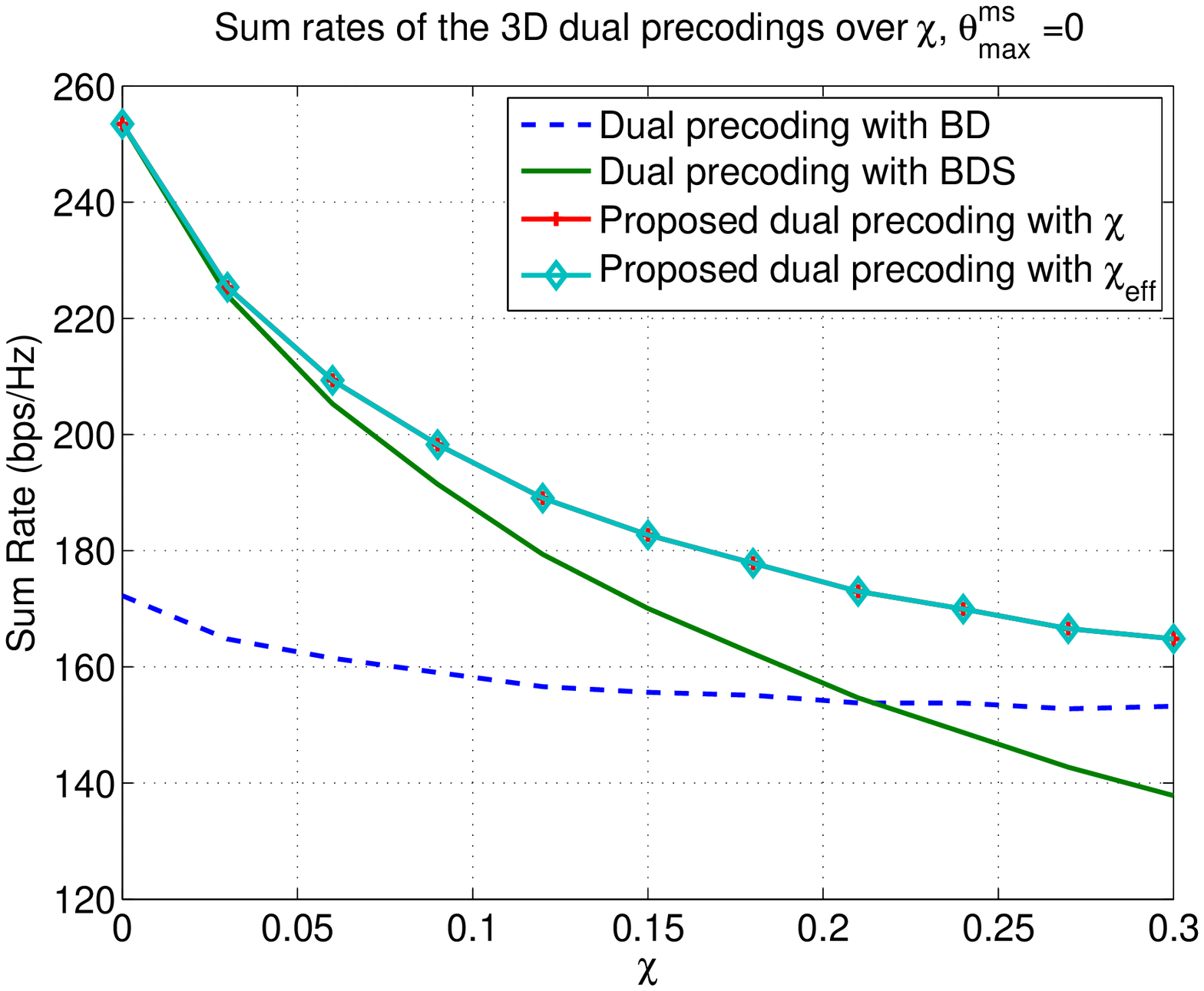}}
 \subfigure[]
  {\includegraphics[height=4.1cm]{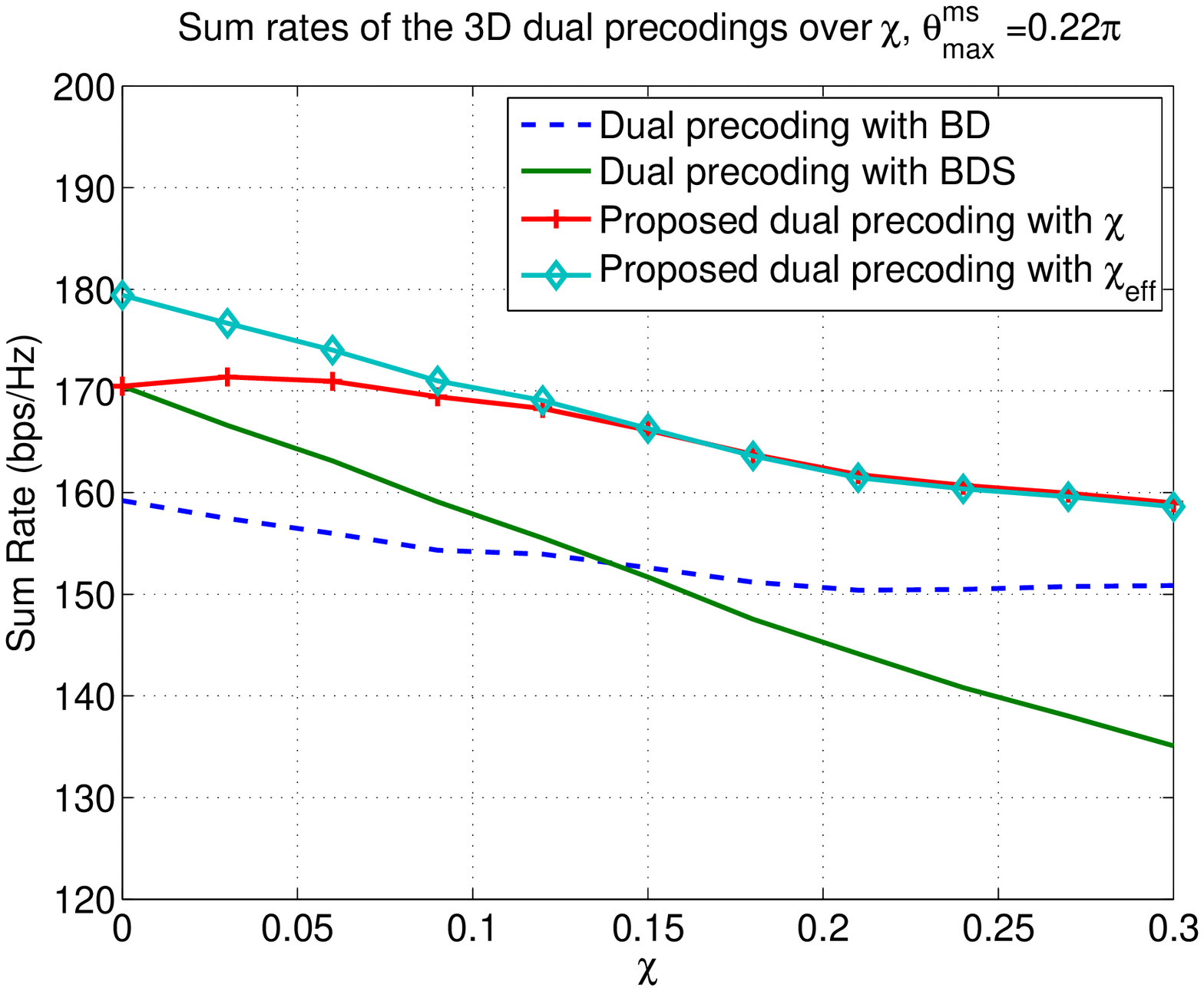}}
 \caption{\label{Fig11_3Dproposed_XPD}  Sum rates of the 3D dual precodings over $\chi$ for (a) $\theta_{\max}^{ms}=0$ (b) $\theta_{\max}^{ms}=0.22\pi$.}
\end{figure}

\section{Conclusion}
\label{sec:conc} In this paper, we have investigated the dual structured linear precoding in the multi-polarized MU Massive MIMO system. In the dual precoding with BD, MSs are grouped based only on the spatial correlation. However, in that with BDS, by subgrouping the co-polarized MSs in the spatially separated groups, we can further reduce the short-term CSI feedback overhead. Based on the random matrix theory, the system performances of dual structured precoding schemes are asymptotically analyzed. From the asymptotic results, we have found that the performance of the dual precoding with BD is insensitive to the XPD, while that of BDS is affected by the XPD parameter. Because the dual precoding with BDS can have more accurate CSIT (across half of the array) than that with BD under the same number of feedback bits, the region of the number of feedback bits where the BDS exhibits better performance than the BD is analytically derived. That is, assuming the CSI is perfectly estimated at MSs, when the number of feedback bits is not large enough to describe the short-term CSIT accurately and the XPD parameter ($\chi$) is small, the dual precoding with BDS exhibits a better performance. Finally, based on that observation, we have proposed a new dual structured precoding/feedback in which the precoding mode is switched between BD and BDS depending on the XPD, spatial correlation, and the number of short-term feedback bits (short-term CSIT quality) and extended it to 3D dual structured precoding.

\useRomanappendicesfalse
\appendices

%\renewcommand\thesection{Appendix \Alph{section}}
%
%\section{Derivation of the inequalities in (\ref{eqn5_prox}) }\label{appndix2}
%
%

\section{Proof of Theorem 2}\label{appndix1}
%based on random matrix theory results \cite{Wagner1,Hachem,HoydisDebbah}, we first derive the asymptotic SINR for two different dual precoding schemes -- dual precoding with i) BD and ii) BDS under the imperfect CSIT and dual-polarized antenna system. Note that the asymptotic inter/intra interferences are evaluated over the polarization domain as well as the spatial domain, which can encounter a more generalized channel environment with a polarization. That is, the asymptotic performance can be analyzed in terms of both the polarization and the spatial correlation Thanks to the structure of the dual polarized antenna elements, we can have a simple asymptotic SINR in Theorem \ref{thm_BD} compared to that for the case of the general antenna covariance matrices \cite{Wagner1}. This gives a useful insight into the behavior of the asymptotic SINR as a function of the polarization parameter in Section \ref{ssec:dualprecoding}

We note that our derivation is based on the derivation of the asymptotic SINR of regularized ZF precoding in spatially correlated MISO broadcasting under the imperfect CSIT \cite{Wagner1}.
 The main difference is that the asymptotic inter/intra interferences are evaluated over the polarization domain as well as the spatial domain. % and that we have the intergroup interferences due to the preprocessing and the co-polarized MSs in the same subgroup have the same spatial correlation as in (\ref{Sys_5_1}).
We first consider the normalization factor ${\xi}_g^2 $ in (\ref{dual_11}). By letting $\Psi \triangleq  \frac{P}{N} tr(\hat{\bar{{\bf H}}}_g^H \hat{\bar{{\bf K}}}_g ^H \hat{\bar{{\bf K}}}_g \hat{\bar{{\bf H}}}_g)$, ${\xi}_g^2 $ can be rewritten as
\begin{eqnarray}\label{AppI_1}
{\xi}_g^2 = \frac{P}{G}\frac{1}{ \Psi }.
\end{eqnarray}
By substituting $\hat{\bar{{\bf K}}}_g$ in (\ref{dual_10_1}), due to the matrix inversion lemma, $\Psi$ can be written as shown at the top of the page.
\begin{figure*}[!t]
\scriptsize
\begin{eqnarray}\label{AppI_2}
\Psi& =&  \frac{P}{N} \sum_{k=1}^{\bar N}\hat{{{\bf h}}}_{gk}^H {\bf B}_g  \left(\hat{\bar{{\bf H}}}_g\hat{\bar{{\bf H}}}_g^H + \bar B \alpha {\bf I}_{\bar B}  \right)^{-2} {\bf B}_g^H  \hat{{{\bf h}}}_{gk}~=~\frac{P}{N\bar B} \sum_{k=1}^{\bar N}\frac{\hat{{{\bf h}}}_{gk}^H {\bf B}_g  \left({\bf A}_{g}^{-k}  + \alpha {\bf I}_{\bar B}  \right)^{-2}{\bf B}_g^H  \hat{{{\bf h}}}_{gk} }{( 1+ \hat{{{\bf h}}}_{gk}^H {\bf B}_g  \left({\bf A}_{g}^{-k} + \alpha {\bf I}_{\bar B}  \right)^{-1} {\bf B}_g ^H \hat{{{\bf h}}}_{gk} )^2}\nonumber\\
&=&\frac{P}{N\bar B}\sum_{p\in\{v,h\}}\sum_{k=1}^{\bar N /2}\frac{\hat{{{\bf g}}}_{gk}^H{\bf R}_{gp}^{1/2} {\bf B}_g  \left({\bf A}_{g}^{-k}  + \bar B \alpha {\bf I}_{\bar B}  \right)^{-2}{\bf B}_g^H  {\bf R}_{gp}^{1/2}\hat{{{\bf g}}}_{gk} }{( 1+ \hat{{{\bf g}}}_{gk}^H {\bf R}_{gp}^{1/2}{\bf B}_g  \left({\bf A}_{g}^{-k} + \alpha {\bf I}_{\bar B}  \right)^{-1} {\bf B}_g ^H{\bf R}_{gp}^{1/2} \hat{{{\bf g}}}_{gk} )^2},
\end{eqnarray}
where ${\bf A}_{g}^{-k} =\frac{1}{\bar B} \hat{\bar{{\bf H}}}_{g}^{-k} (\hat{\bar{{\bf H}}}_{g}^{-k})^H$ and $\hat{\bar{{\bf H}}}_{g}^{-k}=[\hat{\bar{{\bf h}}}_{g1},...,\hat{\bar{{\bf h}}}_{gk-1},\hat{\bar{{\bf h}}}_{gk+1}\hat{\bar{{\bf h}}}_{g\bar N}]$.

\hrulefill \vspace*{2pt}
\end{figure*}
From Lemma 4 in \cite{Wagner1}, we have
\begin{eqnarray}\label{AppI_3}
\!\!\Psi - \frac{P}{N\bar B}\!\! \sum_{p\!\in\!\{v,h\}}\!\sum_{k=1}^{\bar N /2} \!\frac{\frac{1}{\bar B}tr({\bf B}_g ^H{\bf R}_{gp} {\bf B}_g  \left({\bf A}_{g}^{-k}  \!+\! \alpha {\bf I}_{\bar B}  \right)^{-2}  ) }{( 1\!+\! \frac{1}{\bar B}tr({\bf B}_g ^H{\bf R}_{gp}{\bf B}_g  \left({\bf A}_{g}^{-k} \!+\! \alpha {\bf I}_{\bar B}  \right)^{-1} ))^2}\!\! \nonumber\\\overset{M \to \infty}{\longrightarrow}0.\quad\quad\quad\!\!
\end{eqnarray}
Because the covariance matrix of users in the same co-polarized subgroup is equal, from Lemma 6 of \cite{Wagner1}, we have
\begin{eqnarray}\label{AppI_4}
\!\!\Psi - \frac{\bar N P}{2N\bar B}\sum_{p\in \{v,h\}}\frac{\frac{1}{\bar B}tr({\bf B}_g^H{\bf R}_{gp} {\bf B}_g  \left({\bf A}_{g}  + \alpha {\bf I}_{\bar B}  \right)^{-2}    ) }{( 1+ \frac{1}{\bar B}tr({\bf B}_g ^H{\bf R}_{gp}{\bf B}_g  \left({\bf A}_{g} + \alpha {\bf I}_{\bar B}  \right)^{-1} ))^2} \!\! \nonumber\\ \overset{M \to \infty}{\longrightarrow}0,\quad\quad\quad
\end{eqnarray}
where ${\bf A}_{g} =\frac{1}{\bar B} \hat{\bar{{\bf H}}}_{g} \hat{\bar{{\bf H}}}_{g}^H$.
Then, we define $m_{gp}(z) \triangleq \frac{1}{\bar B}tr({\bf B}_g ^H{\bf R}_{gp}{\bf B}_g \left({\bf A}_{g} - z {\bf I}_{\bar B}  \right)^{-1} )$, which is analogous to (\ref{Asym_1}) in Theorem 1 by setting ${\bf S}={\bf 0}$ and ${\bf Q} = {\bf B}_g ^H{\bf R}_{gp}{\bf B}_g $($= \bar{\bf R}_{gp}$ from (\ref{dual_9})). The trace in the denominator of (\ref{AppI_4}) is then equal to $m_{gp}(-\alpha)$ and, from Theorem 1,
\begin{eqnarray}\label{AppI_6}
\!\frac{1}{\bar B}tr({\bf B}_g ^H{\bf R}_{gp}{\bf B}_g  \left({\bf A}_{g}\! + \alpha {\bf I}_{\bar B}  \right)^{-1} ) - \frac{1}{\bar B}tr(\bar{\bf R}_{gp}{\bf T}_g) \overset{M \to \infty}{\longrightarrow}0,
\end{eqnarray}
where ${\bf T}_g$ is given by (\ref{Asym_6_1}) from (\ref{Asym_2}) and (\ref{Asym_3}). Furthermore, the trace in the numerator of (\ref{AppI_4}) is the derivative of $ m_{gp}(z)$ at $z= -\alpha$, i.e., $m_{gp}'(-\alpha)$. Therefore,
\begin{eqnarray}\label{AppI_7}
\frac{1}{\bar B}tr({\bf B}_g^H{\bf R}_{gp} {\bf B}_g  \left({\bf A}_{g}  + \alpha {\bf I}_{\bar B}  \right)^{-2} ) - \frac{1}{\bar B}tr(\bar{\bf R}_{gp}{\bf T}'_g) \overset{M \to \infty}{\longrightarrow}0,
\end{eqnarray}
where
\begin{eqnarray}\label{AppI_8}
{\bf T}'_g = {\bf T}_g\left( \frac{\bar N}{2\bar B}\sum_{p\in\{v,h\}}\frac{\bar{\bf R}_{gp} m'_{gp}}{(1+m_{gp}(-\alpha))^2}  +{\bf I}_{\bar B} \right){\bf T}_g,
\end{eqnarray}
\begin{eqnarray}\label{AppI_8_1}
m'_{gp} = \frac{1}{\bar B}tr(\bar{\bf R}_{gp} {\bf T}'_g) .
\end{eqnarray}
Then, by putting ${\bf T}'_g$ in (\ref{AppI_8}) into (\ref{AppI_8_1}), ${\bf m}'_g$ can be obtained as in (\ref{Asym_6_3}) and $\Psi$ is converged as
\begin{eqnarray}\label{AppI_9}
\Psi - \frac{P \bar N}{2N\bar B} \sum_{p\in\{v,h\}}\frac{ m'_{gp}}{( 1+ m_{gp}^o)^2} \overset{M \to \infty}{\longrightarrow}0.
\end{eqnarray}

For the signal power component $\frac{P}{N}|{\bf h}_{gk}^H{\bf B}_g  \hat{\bar{{\bf K}}}_g \hat{\bar{{\bf h}}}_{gk} |^2 $ and the intra-group interference component $\!\frac{P}{N}\sum_{\!j\!\neq\! k\!}|{\bf h}_{gk}^H{\bf B}_g  \hat{\bar{{\bf K}}}_g \hat{\bar{{\bf h}}}_{gj} |^2\!$ in (\ref{dual_12}), because ${\bf G}_g$ and ${\bf Z}_g$ in (\ref{Sys_7}) are independent, by taking a similar approach described in (\ref{AppI_2}) and (\ref{AppI_3}) (See also Appendix II.B and C in \cite{Wagner1}), we can have the following relations:
\begin{eqnarray}\label{AppI_10}
|{\bf h}_{gk}^H{\bf B}_g  \hat{\bar{{\bf K}}}_g \hat{\bar{{\bf h}}}_{gk} |^2 -  \frac{  (1-\tau^2)(m_{gp}^o)^2 }{( 1+ m_{gp}^o)^2} \overset{M \to \infty}{\longrightarrow}0,
\end{eqnarray}
\begin{eqnarray}\label{AppI_11}
\sum_{j\neq k}|{\bf h}_{gk}^H{\bf B}_g  \hat{\bar{{\bf K}}}_g \hat{\bar{{\bf h}}}_{gj} |^2\! - \!\frac{\Upsilon_{ggp}^{k}  (1\!- \! \tau^2(1\!-\!(1\!+\!m_{gp}^o)^2)) }{( 1\!+\! m_{gp}^o)^2} \overset{M \to \infty}{\longrightarrow}0,\!\!
\end{eqnarray}
where
\begin{eqnarray}\label{AppI_12}
\Upsilon_{ggp}^{k} = \frac{1}{\bar B}\sum_{j\neq k} \hat{\bf g}_{gj}^H{\bf R}_{gp}^{1/2}{\bf B}_g  \hat{\bar{{\bf K}}}_g {\bf B}_g^H {\bf R}_{gp}{\bf B}_g \hat{\bar{{\bf K}}}_g {\bf B}_g^H {\bf R}_{gp}^{1/2}\hat{\bf g}_{gj}\nonumber\\
+\frac{1}{\bar B}\sum_{j=1}^{\bar N/2} \hat{\bf g}_{gj}^H{\bf R}_{gq}^{1/2}{\bf B}_g  \hat{\bar{{\bf K}}}_g {\bf B}_g^H {\bf R}_{gp}{\bf B}_g \hat{\bar{{\bf K}}}_g {\bf B}_g^H {\bf R}_{gq}^{1/2}\hat{\bf g}_{gj}.
\end{eqnarray}
Note that by taking a similar approach described in (\ref{AppI_2}) and (\ref{AppI_3}),
\begin{eqnarray}\label{AppI_13}
\Upsilon_{ggp}^{k} - \frac{\bar N/2-1}{\bar B}\frac{\frac{1}{\bar B}tr (\bar{\bf R}_{gp} ({\bf A}_g +\alpha {\bf I}_{\bar B} )^{-1} \bar{\bf R}_{gp} ({\bf A}_g +\alpha {\bf I}_{\bar B} )^{-1}) }{(1+\frac{1}{\bar B} tr (\bar{\bf R}_{gp} ({\bf A}_g +\alpha {\bf I}_{\bar B} )^{-1}))^2 } \nonumber\\ -\frac{\bar N}{2\bar B}\frac{\frac{1}{\bar B}tr (\bar{\bf R}_{gq} ({\bf A}_g \!+\!\alpha {\bf I}_{\bar B} )^{-1} \bar{\bf R}_{gp} ({\bf A}_g +\alpha {\bf I}_{\bar B} )^{-1}) }{(1+\frac{1}{\bar B} tr (\bar{\bf R}_{gq} ({\bf A}_g \!+\!\alpha {\bf I}_{\bar B} )^{-1}))^2 }\overset{M \to \infty}{\longrightarrow}0.\!\!
\end{eqnarray}
By defining $m_{ggpq}(z) \triangleq \frac{1}{\bar B}tr(\bar{\bf R}_{gq} \left({\bf A}_{g} +\alpha {\bf I}_{\bar B}- z \bar{\bf R}_{gp}  \right)^{-1} )$, the trace in the numerator of (\ref{AppI_13}) is the derivative of $m_{ggpq}(z)$ at $z= 0$. Furthermore, by using Theorem 1,
\begin{eqnarray}\label{AppI_14}
m_{ggpq}(z) - \frac{1}{\bar B}tr(\bar{\bf R}_{gq}{\bf T}_{gp}(z)) \overset{M \to \infty}{\longrightarrow}0,
\end{eqnarray}
where ${\bf T}_{gp}(z) = (\frac{\bar N}{2\bar B}\sum_{q\in\{v,h\}} \frac{\bar{\bf R}_{gq}}{1+m_{ggqp}(z)} + \alpha{\bf I}_{\bar B} - z\bar{\bf R}_{gp})^{-1}$. Accordingly, we have
\begin{eqnarray}\label{AppI_16}
\!{\bf T}'_{gp}(z) \!= \!{\bf T}_{gp}(z)(\frac{\bar N}{2\bar B} \!\sum_{q\in\{v,h\}} \frac{\bar{\bf R}_{gq} m'_{ggqp}(z)}{(1\!+\!m_{ggqp}(z))^2}  \!+ \!\bar{\bf R}_{gp}){\bf T}_{gp}(z).\!
\end{eqnarray}
Note that ${\bf T}_{gp}(0) = {\bf T}_{g}$ in (\ref{Asym_6_1}) and from (\ref{AppI_14}), $ m_{ggpq}(0) -m_{gq}^o \overset{M \to \infty}{\longrightarrow}0$. Together with $m'_{ggqp}(z) - \frac{1}{\bar B}tr(\bar{\bf R}_{gq}{\bf T}'_{gp}(z)) \overset{M \to \infty}{\longrightarrow}0$, ${\bf m}'_{ggp}$ can be obtained as (\ref{Asym_6_3}). Accordingly,
\begin{eqnarray}\label{AppI_17}
\Upsilon_{ggp}^{k}  -\Upsilon_{ggp}^{o} \overset{M \to \infty}{\longrightarrow}0,
\end{eqnarray}
where $\Upsilon_{ggp}^{o}$ is defined in (\ref{Asym_6_2}).

Now let us consider the inter-group interference component $\Upsilon_{glp}^{k} \triangleq \sum_{j} |{\bf h}_{gk}^H{\bf B}_l  \hat{\bar{{\bf K}}}_l {\bf B}_l^H \hat{{\bf h}}_{lj} |^2$ in the denominator of (\ref{dual_12}) for $l\neq g$. Then, $\Upsilon_{glp}^{k}$ can be rewritten as
\begin{eqnarray}\label{AppI_18}
\!\!\!\Upsilon_{glp}^{k}\! =\!  \frac{1}{\bar B}\!\sum_{j\!=\!1}^{\bar N/2} \!\sum_{q\!\in\!\{v,h\}}\hat{\bf g}_{lj}^H{\bf R}_{lq}^{1/2}{\bf B}_l  \hat{\bar{{\bf K}}}_l {\bf B}_l^H {\bf R}_{gp}{\bf B}_l \hat{\bar{{\bf K}}}_l {\bf B}_l^H {\bf R}_{lq}^{1/2}\hat{\bf g}_{lj}.\!\!
\end{eqnarray}
Note that (\ref{AppI_18}) has a similar form with (\ref{AppI_12}). Therefore, similarly done in (\ref{AppI_13}), we have
\begin{eqnarray}\label{AppI_19}
\!\!\Upsilon_{glp}^{k} \!- \!\frac{\bar N}{2\bar B}\!\!\sum_{q\!\in\!\{v,h\}}\!\!\frac{\frac{1}{\bar B}\!tr (\bar{\bf R}_{lq} ({\bf A}_l \!+\!\alpha {\bf I}_{\bar B}\! )^{-1}\!  {\bf B}_l^H {\bf R}_{gp}{\bf B}_l \!({\bf A}_l\! +\!\alpha {\bf I}_{\bar B}\! )^{-1}\!) }{(1v+\!\frac{1}{\bar B} tr (\!\bar{\bf R}_{lq} ({\bf A}_l \!+\!\alpha {\bf I}_{\bar B} )^{-1}\!)\!)^2 }\!\!\!\nonumber\\\!\! \overset{M \to \infty}{\longrightarrow}0.\quad\quad\quad\!\!\!
\end{eqnarray}
Furthermore, by following the steps similarly done in (\ref{AppI_14})-(\ref{AppI_16}), we can obtain
\begin{eqnarray}\label{AppI_20}
\Upsilon_{gl}^{k}  -\Upsilon_{gl}^{o} \overset{M \to \infty}{\longrightarrow}0,
\end{eqnarray}
where $\Upsilon_{glp}^{o}$ is defined in (\ref{Asym_6_2}).
By substituting (\ref{AppI_9}), (\ref{AppI_10}), (\ref{AppI_11}), (\ref{AppI_17}), and (\ref{AppI_20}) into (\ref{dual_12}), we can have (\ref{Asym_4}).

\section{Proof of Proposition 1}\label{appndix2}

Let ${\bf R}_g (\chi) = (1+ \chi)\left[\begin{array}{cc}{\bf R}_g^s & {\bf 0}\\{\bf 0} &{\bf R}_g^s \end{array}\right]$ from (\ref{Sys_5_1}). Then, we have ${\bf R}_g (\chi) = (1+ \chi){\bf R}_g (0)$ (respectively, $\bar{\bf R}_g (\chi) = (1+ \chi)\bar{\bf R}_g (0)$) and, from Corollary \ref{cor_BD}, $\gamma_{gpk}^{BD,o}(\chi)$ can be evaluated by using (\ref{Asym_8}) with $\frac{1}{2}\bar{\bf R}_g (\chi)$.
From (\ref{Asym_9_1}), we define
\begin{eqnarray}\label{Asym_16}
{\bf T}_g(\chi) &=&  \left(\frac{\bar N}{2 \bar B}\frac{\bar{\bf R}_g(\chi)}{1+m_g^o(\chi)} +\alpha{\bf I}_{\bar B} \right)^{-1},\\ \label{Asym_16_2}
m_g^o(\chi) &=& \frac{1}{2\bar B}tr(\bar{\bf R}_g(\chi) {\bf T}_g(\chi)).
\end{eqnarray}
First, let us consider the high SNR regime (i.e., small $\alpha$). Because
\begin{eqnarray}\label{Asym_16_1}
tr({\bf R}(\beta {\bf R} + \alpha {\bf I})^{-1}) = \frac{1}{\beta}tr({\bf R}({\bf R} + \alpha/{\beta} {\bf I})^{-1})\nonumber\\
\approx \frac{1}{\beta}tr({\bf R}\left({\bf R} + \alpha {\bf I} \right)^{-1}),
\end{eqnarray}
for any Hermitian nonnegative definite matrix ${\bf R}$ and small $\alpha$, by substituting (\ref{Asym_16}) into (\ref{Asym_16_2}), we have
\begin{eqnarray}\label{Asym_17}
\!\!m_g^o(\chi)\! \!&\!\!\!=\!\!\!& \!\! \frac{1}{2\bar B}tr((1\!+\!\chi)\bar{\bf R}_g(0)\left(\frac{\bar N}{2 \bar B}\frac{(1\!+\!\chi)\bar{\bf R}_g(0)}{1+m_g^o(\chi)} \!+\!\alpha{\bf I}_{\bar B} \right)^{-1}),\!\!\nonumber\\\!\!\!&\!\!\!\approx\!\!\!&  \frac{1}{2\bar B}tr(\bar{\bf R}_g(0)\left(\frac{\bar N}{2 \bar B}\frac{\bar{\bf R}_g(0)}{1+m_g^o(\chi)} +\alpha{\bf I}_{\bar B} \right)^{-1}),\!\!
\end{eqnarray}
which implies that
\begin{eqnarray}\label{Asym_17_1}
m_g^o(\chi) = m_g^o(0).
\end{eqnarray}
Furthermore, by substituting (\ref{Asym_17_1}) into (\ref{Asym_16}),
\begin{eqnarray}\label{Asym_18}
{\bf T}_g(\chi) &\!=\!&  \frac{1}{1+\chi } \left(\frac{\bar N}{2 \bar B}\frac{\bar{\bf R}_g(0)}{1+m_g^o(0)} +\frac{\alpha}{1+\chi}{\bf I}_{\bar B} \right)^{-1}.\!
\end{eqnarray}
By using (\ref{Asym_16_1}), (\ref{Asym_17_1}) and (\ref{Asym_18}), we can also derive
\begin{eqnarray}\label{Asym_19}
\!\!m'_g (\chi ) \!\approx\! \frac{1}{1\!+\! \chi}  m'_g (0),~m'_{gg}(\chi)\! \approx \!m'_{gg}(0),~m'_{gl}(\chi)\! \approx\! m'_{gl}(0),\!
\end{eqnarray}
which implies that
\begin{eqnarray}\label{Asym_20}
\!\!\Psi_g^o(\chi) \!\approx\! \frac{1}{1\!+\!\chi}\Psi_g^o(0),~ \Upsilon_{gg}^o(\chi) \!\approx \!\Upsilon_{gg}^o(0),\nonumber\\\Upsilon_{gl}^o(\chi) \!\approx \!\Upsilon_{gl}^o(0),~{\xi}_g^o(\chi) \!\approx\!(1\!+\!\chi) {\xi}_g^o(0).\!\!
\end{eqnarray}
Then, based on (\ref{Asym_17_1}) and (\ref{Asym_20}) together with (\ref{Asym_8}), $
\gamma_{gpk}^{BD,o}(\chi) \approx \gamma_{gpk}^{BD,o}(0)$.

For low SNR regime (i.e., large $\alpha$), in (\ref{Asym_16}), ${\bf T}_g(\chi) \approx \frac{1}{\alpha}{\bf I}_{\bar B}$ and accordingly, $m_g^o(\chi) \approx (1+\chi)m_g^o(0) \ll 1$. Furthermore,
\begin{eqnarray}\label{Asym_22}
\!\!m'_g (\chi ) \!\approx\! (1\!+ \!\chi)  m'_g (0),~ m'_{gg}(\chi) \!\approx\! (1\!+ \!\chi)^2m'_{gg}(0),\nonumber\\ m'_{gl}(\chi) \!\approx\! (1\!+\! \chi)^2m'_{gl}(0),\!\!
\end{eqnarray}
and $ \Psi_g^o(\chi) \approx (1+\chi)\Psi_g^o(0)$, $\Upsilon_{gg}^o(\chi) \approx  (1+ \chi)^2\Upsilon_{gg}^o(0)$, $\Upsilon_{gl}^o(\chi) \approx  (1+ \chi)^2\Upsilon_{gl}^o(0)$, and ${\xi}_g^o(\chi) \approx \frac{1}{1+\chi} {\xi}_g^o(0)$.
Accordingly, in the low SNR regime, we can have that $\gamma_{gpk}^{BD,o}(\chi) \approx \gamma_{gpk}^{BD,o}(0)$.
Accordingly, in the low SNR regime, we can have that $\gamma_{gpk}^{BD,o}(\chi) \approx \gamma_{gpk}^{BD,o}(0)$.

\section{Proof of Proposition 2}\label{appndix3}
From (\ref{Sys_5_1}) and (\ref{dual_15}), we have $\bar{\bf R}_{gv}(\chi) = \bar{\bf R}_{gh}(\chi)= ({\bf B}_g^s)^H {\bf R}_g^s {\bf B}_g^s$.
That is, $\bar{\bf R}_{gv}(\chi)$ and $\bar{\bf R}_{gh}(\chi)$ are independent of $\chi$ and expressed as $\bar{\bf R}_{gp}(\chi)= \bar{\bf R}_{gp}(0)$. Accordingly, from (\ref{Asym_13_1}) and (\ref{Asym_14_1}),
\begin{eqnarray}\label{Asym_24}
m_{gp}^o(\chi) = m_{gp}^o(0),\quad {\bf T}_{gp}(\chi) = {\bf T}_{gp}(0), \nonumber \\ m'_{gp}(\chi)= m'_{gp}(0), \quad m'_{ggpp}(\chi)= m'_{ggpp}(0).
\end{eqnarray}
Because ${\bf B}_{lq}^H {\bf R}_{gp}(\chi){\bf B}_{lq} = \left\{\begin{array}{c} {\bf B}_{lq}^H {\bf R}_{gp}(0){\bf B}_{lq} {\text{ for }} q=p \\  \chi {\bf B}_{lq}^H {\bf R}_{gq}(0){\bf B}_{lq} {\text{ for }} q\neq p \end{array}    \right.$,
from (\ref{Asym_13_1}), $m'_{glpq}(\chi) = \left\{\begin{array}{c} m'_{glpq}(0) {\text{ for }} q=p \\  \chi m'_{glqq}(0) {\text{ for }} q\neq p \end{array}    \right. .$
Therefore, from (\ref{Asym_12}),
\begin{eqnarray}\label{Asym_27}
\gamma_{gpk}^{BDS,o}(\chi) =  \frac{\frac{P}{N}({\xi}_{gp}^o)^2 (1-\tau^2) (m_{gp}^o)^2}{IN_{gpk}^{o}(\chi)},
\end{eqnarray}
where $IN_{gpk}^{o}(\chi) $ is given as shown at the top of the page.
\begin{figure*}[!t]
\scriptsize
\begin{eqnarray}\label{Asym_28}
&IN_{gpk}^{o}(\chi) =  ({\xi}_{gp}^o(0))^2 \Upsilon_{ggpp}^o(0) (1-\tau^2(1-(1+ m_{gp}^o(0))^2))\nonumber &\\
&+(1+ \chi({\xi}_{gp}^o(0))^2\Upsilon_{ggpp}^o(0)  +\sum_{l \neq g}(1+\chi)({\xi}_{lp}^o(0))^2\Upsilon_{glpp}^o(0) )(1+ m_{gp}^o(0))^2,&\nonumber\\
&=IN_{gpk}^{o}(0)  +\chi( ({\xi}_{gp}^o(0))^2\Upsilon_{ggpp}^o(0)  +\sum_{l \neq g}({\xi}_{lp}^o(0))^2\Upsilon_{glpp}^o(0) )(1+ m_{gp}^o(0))^2.&
\end{eqnarray}
\hrulefill \vspace*{2pt}
\end{figure*}

Because of the preprocessing with BD, $\Upsilon_{ggpp}^o(0) \gg \Upsilon_{glpp}^o(0)$ for $g \neq l$ and $IN_{gpk}^{o}(\chi) \approx IN_{gpk}^{o}(0)(1+ c_{0,gpk} \chi)$,
where $ c_{0,gpk} = \frac{({\xi}_{gp}^o(0))^2\Upsilon_{ggpp}^o(0) (1+ m_{gp}^o(0))^2}{({\xi}_{gp}^o(0))^2\Upsilon_{ggpp}^o(0)(\tau^2((1+ m_{gp}^o(0))^2-1) +1)  +(1+ m_{gp}^o(0))^2}$. By averaging $c_{0,gpk}$ over $g$, $p$, we can have (\ref{Asym_22_1}).

%\renewcommand\baselinestretch{0.9}
%\section*{Acknowledgments}
%% optional entry into table of contents (if used)
%%\addcontentsline{toc}{section}{Acknowledgment}
%This research was supported by the KCC (Korea Communications
%Commission), Korea, under the R\&D program supervised by the KCA
%(Korea Communications Agency) (KCA-2011- 10921-01305).

% References should be produced using the bibtex program from suitable
% BiBTeX files (here: strings, refs, manuals). The IEEEbib.bst bibliography
% style file from IEEE produces unsorted bibliography list.
% -------------------------------------------------------------------------
%\renewcommand\baselinestretch{1.2}
\bibliographystyle{IEEEtran}
\bibliography{IEEEabrv,myref}

\begin{thebibliography}{10}
\providecommand{\url}[1]{#1}
\csname url@rmstyle\endcsname
\providecommand{\newblock}{\relax}
\providecommand{\bibinfo}[2]{#2}
\providecommand\BIBentrySTDinterwordspacing{\spaceskip=0pt\relax}
\providecommand\BIBentryALTinterwordstretchfactor{4}
\providecommand\BIBentryALTinterwordspacing{\spaceskip=\fontdimen2\font plus
\BIBentryALTinterwordstretchfactor\fontdimen3\font minus
  \fontdimen4\font\relax}
\providecommand\BIBforeignlanguage[2]{{%
\expandafter\ifx\csname l@#1\endcsname\relax
\typeout{** WARNING: IEEEtran.bst: No hyphenation pattern has been}%
\typeout{** loaded for the language `#1'. Using the pattern for}%
\typeout{** the default language instead.}%
\else
\language=\csname l@#1\endcsname
\fi
#2}}

\bibitem{Marzetta}
T.~L. Marzetta, ``Noncooperative cellular wireless with unlimited numbers of
  base station antennas,'' \emph{{IEEE} Trans. Wireless Commun.}, vol.~9,
  no.~11, pp. 3590--3600, Nov. 2010.

\bibitem{Marzetta2}
F.~Rusek, D.~Persson, B.~K. Lau, E.~G. Larsson, T.~L. Marzetta, O.~Edfors, and
  F.~Tufvesson, ``Scaling up {MIMO}: opportunities and challenges with very
  large arrays,'' \emph{{IEEE} Signal Processing Mag.}, vol.~30, no.~1, pp.
  40--60, Jan. 2013.

\bibitem{Huh}
H.~Huh, G.~Caire, H.~C. Papadopoulos, and S.~A. Ramprashad, ``Achieving
  ``massive {MIMO}'' spectral efficiency with a not-so-large number of
  antennas,'' \emph{{IEEE} Trans. Wireless Commun.}, vol.~11, no.~9, pp.
  3226--3239, Sept. 2012.

\bibitem{HoydisDebbah}
J.~Hoydis, S.~Brink, and M.~Debbah, ``Massive {MIMO} in the {UL/DL} of cellular
  networks: How many antennas do we need?'' \emph{{IEEE} J. Select. Areas
  Commun.}, vol.~31, no.~2, pp. 160--171, Feb. 2013.

\bibitem{MohLarsson}
S.~K. Mohammed and E.~G. Larsson, ``Single-user beamforming in large-scale
  {MISO} systems with per-antenna constant-envelope constraints: the doughnut
  channel,'' \emph{{IEEE} Trans. Wireless Commun.}, vol.~11, no.~11, pp.
  3992--4005, Nov. 2012.

\bibitem{MohLarsson2}
------, ``Per-antenna constant envelope precoding for large multi-user {MIMO}
  systems,'' \emph{{IEEE} Trans. Commun.}, vol.~61, no.~3, pp. 1059--1071, Mar.
  2013.

\bibitem{Wagner1}
S.~Wagner, R.~Couillet, M.~Debbah, and D.~T.~M. Slock, ``Large system analysis
  of linear precoding in correlated {MISO} broadcast channels under limited
  feedback,'' \emph{{IEEE} Trans. Inform. Theory}, vol.~58, no.~7, pp.
  4509--4537, July 2012.

\bibitem{Caire}
A.~Adhikary, J.~Nam, J.~Ahn, and G.~Caire, ``Joint spatial division and
  multiplexing-the large-scale array regime,'' \emph{{IEEE} Trans. Inform.
  Theory}, vol.~59, no.~10, pp. 6441--6463, Oct. 2013.

\bibitem{ZhangWen}
J.~Zhang, C.~Wen, S.~Jin, X.~Gao, and K.~Wong, ``Large system analysis of
  cooperative multi-cell downlink transmission via regularized channel
  inversion with imperfect {CSIT},'' \emph{{IEEE} Trans. Wireless Commun.},
  vol.~12, no.~10, pp. 4801--4813, Oct. 2013.

\bibitem{Clerckx_book}
B.~Clerckx and C.~Oestges, \emph{{MIMO} Wireless Networks: Channels, Techniques
  and Standards for Multi-Antenna, Multi-User and Multi-Cell Systems}.\hskip
  1em plus 0.5em minus 0.4em\relax Oxford, UK: Academic Press (Elsevier), 2013.

\bibitem{LimYooClerckx}
C.~Lim, T.~Yoo, B.~Clerckx, B.~Lee, and B.~Shim, ``Recent trend of multiuser
  {MIMO} in {LTE Advanced},'' \emph{{IEEE} Commun. Mag.}, vol.~51, no.~3, pp.
  127--135, Mar. 2013.

\bibitem{BrunoKimCon}
B.~Clerckx, G.~Kim, and S.~Kim, ``Correlated fading in broadcast {MIMO}
  channels: curse or blessing?'' in \emph{Proc. {IEEE} Global
  Telecommunications Conference, 2008}, Nov. 2008.

\bibitem{ClerckxCraeye}
B.~Clerckx, C.~Craeye, D.~Vanhoenacker-Janvier, and C.~Oestges, ``Impact of
  antenna coupling on 2x2 {MIMO} communications,'' \emph{{IEEE} Trans. Veh.
  Technol.}, vol.~56, no.~3, pp. 1009--1018, May 2007.

\bibitem{KimBruno}
T.~Kim, B.~Clerckx, D.~J. Love, and S.~J. Kim, ``Limited feedback beamforming
  systems for dual-polarized {MIMO} channels,'' \emph{{IEEE} Trans. Wireless
  Commun.}, vol.~9, no.~11, pp. 3425--3439, Nov. 2010.

\bibitem{OestgesBruno}
C.~Oestges, B.~Clerckx, M.~Guillaud, and M.~Debbah, ``Dual-polarized wireless
  communications: from propagation models to system performance evaluation,''
  \emph{{IEEE} Trans. Wireless Commun.}, vol.~7, no.~10, pp. 4019--4031, Oct.
  2008.

\bibitem{Sulonen}
K.~Sulonen, P.~Suvikunnas, L.~Vuokko, J.~Kivinen, and P.~Vainikainen,
  ``Comparison of {MIMO} antenna configurations in picocell and microcell
  environments,'' \emph{{IEEE} J. Select. Areas Commun.}, vol.~21, no.~5, pp.
  703--712, June 2003.

\bibitem{DongHeath}
L.~Dong, H.~Choo, J.~R.~W.~Heath, and H.~Ling, ``Simulation of {MIMO} channel
  capacity with antenna polarization diversity,'' \emph{{IEEE} Trans. Wireless
  Commun.}, vol.~4, no.~4, pp. 1869--1873, July 2005.

\bibitem{Boccardi}
F.~Boccardi, B.~Clerckx, A.~Ghosh, E.~Hardouin, G.~J{\"{o}}ngren, K.~Kusume,
  E.~Onggosanusi, and Y.~Tang, ``Multiple antenna techniques in
  {LTE}-advanced,'' \emph{{IEEE} Commun. Mag.}, vol.~50, no.~2, pp. 114--121,
  Mar. 2012.

\bibitem{Rao_book}
K.~R. Rao and P.~C. Yip, \emph{The transform and data compression
  handbook}.\hskip 1em plus 0.5em minus 0.4em\relax Florida, USA: CRC Press
  LLC, 2001.

\bibitem{LiJi}
Y.~Li, X.~Ji, D.~Liang, and Y.~Li, ``Dynamic beamforming for three-dimensional
  {MIMO} technique in {LTE}-advanced networks,'' \emph{to be apperaed in
  International Journal of Antennas and Propagation}, 2013.

\bibitem{Coldrey}
M.~Coldrey, ``Modeling and capacity of polarized {MIMO} channels,'' in
  \emph{Proc. {IEEE} Vehicular Technology Conference, 2008}, May 2008, pp.
  440--444.

\bibitem{Asplund}
H.~Asplund, J.~Berg, F.~Harrysson, J.~Medbo, and M.~Riback, ``Propagation
  characteristics of polarized radio waves in cellular communications,'' in
  \emph{Proc. {IEEE} Vehicular Technology Conference, 2007}, Sept. 2007, pp.
  839--843.

\bibitem{N_Jindal2}
N.~Jindal, ``{MIMO} broadcast channels with finite-rate feedback,''
  \emph{{IEEE} Trans. Inform. Theory}, vol.~52, no.~11, pp. 5045--5060, Nov.
  2006.

\bibitem{YinGesbert}
H.~Yin, D.~Gesbert, M.~Filippou, and Y.~Liu, ``A coordinated approach to
  channel estimation in large-scale multiple-antenna systems,'' \emph{{IEEE} J.
  Select. Areas Commun.}, vol.~31, no.~2, pp. 264--273, Feb. 2013.

\bibitem{HwangClerckx}
D.~Hwang, B.~Clerckx, and G.~Kim, ``Regularized channel inversion with
  quantized feedback in downlink multiuser channels,'' \emph{{IEEE} Trans.
  Wireless Commun.}, vol.~8, no.~12, pp. 5785--5789, Dec. 2009.

\bibitem{Tse_book}
D.~Tse and P.~Viswanath, \emph{Fundamental of Wireless Communication},
  1st~ed.\hskip 1em plus 0.5em minus 0.4em\relax Cambridge: Cambridge Univ.
  Press, 2005.

\bibitem{Hachem}
W.~Hachem, P.~Loubaton, and J.~Najim, ``Deterministic equivalents for certain
  functionals of large random matrices,'' \emph{Annals of Applied Probability},
  vol.~17, no.~3, pp. 875--930, June 2007.

\bibitem{HJoungYook}
H.~Joung, H.~Jo, C.~Mun, and J.~Yook, ``Capacity loss due to
  polarization-mismatch and space-correlation on {MISO} channel,'' \emph{{IEEE}
  Trans. Wireless Commun.}, vol.~13, no.~4, pp. 2124--2136, Apr. 2014.

\bibitem{Shiu_book}
D.~S. Shiu, \emph{Wireless communication using dual antenna array}.\hskip 1em
  plus 0.5em minus 0.4em\relax Massachusetts, USA: Kluwer Academic Publishers,
  2000.

\bibitem{ShiuFoschini}
D.~S. Shiu, G.~J. Foschini, M.~J. Gans, and J.~M. Kahn, ``Fading correlation
  and its effect on the capacity of multi-element antenna systems,'' in
  \emph{Proc. {IEEE} International Conference on Universal Personal
  Communications, 1998}, Oct. 1998, pp. 429--433.

\end{thebibliography}

\end{document}